\documentclass[12pt,a4paper,twoside,english,american]{scrartcl}
\usepackage{helvet}
\usepackage{mathtools}
\usepackage{url}
\usepackage[T1]{fontenc}
\usepackage[latin9]{inputenc}
\usepackage{amsmath}
\usepackage{setspace}
\usepackage{amssymb}
\usepackage{enumerate}
\usepackage{pstcol}

\makeatletter

 \usepackage{amsthm}
 \theoremstyle{plain}
\newtheorem{thm}{Theorem}[section]
  \theoremstyle{plain}
  \newtheorem*{thm*}{Theorem}
  \theoremstyle{plain}
  
  \theoremstyle{remark}
  \newtheorem{rem}[thm]{Remark}
  \theoremstyle{definition}
  \newtheorem{defn}[thm]{Definition}
 \theoremstyle{definition}
 \newtheorem*{defn*}{Definition}
  \theoremstyle{plain}
  \newtheorem{lem}[thm]{Lemma}
 \theoremstyle{definition}
  
  \theoremstyle{remark}

\usepackage[T1]{fontenc}    

\usepackage{a4wide}        
\usepackage{tu-preprint}
\usepackage{amsmath}

\usepackage{euscript}
\usepackage{mathabx}
\allowdisplaybreaks[1] 
\usepackage{amsfonts}
\newenvironment{keywords}{ \noindent\footnotesize\textbf{Keywords and phrases:}}{}

\newenvironment{class}{\noindent\footnotesize\textbf{Mathematics subject classification 2010:}}{}

\usepackage{color,tu-preprint}

\newcommand*{\dive}{\operatorname{div}}

\newcommand*{\curl}{\operatorname{curl}}

\newcommand*{\grad}{\operatorname{grad}}

\newcommand*{\ii}{\mathrm{i}}
\renewcommand*{\i}{\mathrm{i}}
\renewcommand*{\d}{\mathrm{d}}

\newcommand{\oi}[2]{\left]#1,#2 \right[}

\newcommand{\lci}[2]{\left[#1,#2 \right[}
\newcommand{\rci}[2]{\left]#1,#2 \right]}
\newcommand{\Max}{\mathrm{Max}}
\newcommand{\Dac}{\mathrm{Dac}}
\newcommand{\Nac}{\mathrm{Nac}}
\newcommand{\ac}{\mathrm{ac}}
\newcommand{\E}{\mathcal{E}}

\DeclareMathAccent{\Circ}{\mathalpha}{operators}{"17}
\newcommand{\interior}[1]{\Circ{#1}}
\renewcommand{\Im}{\operatorname{\mathfrak{Im}}}
\renewcommand{\Re}{\operatorname{\mathfrak{Re}}}

\theoremstyle{plain}

\newtheorem{Sa*}[section]{Theorem}

\newtheorem{Le*}[section]{Lemma}

\newtheorem{Fo*}[subsection]{Corollary}

\newtheorem{Prop*}[section]{Proposition}

\theoremstyle{definition}

\theoremstyle{remark}

 \numberwithin{equation}{section}

\DeclareMathOperator{\TextRe}{Re}
\DeclareMathOperator{\TextIm}{Im}
\renewcommand{\Re}{\TextRe}
\renewcommand{\Im}{\TextIm}

\newcommand{\C}{\mathbb{C}}

\newcommand{\ben}{\begin{enumerate}[(i)]}
\newcommand{\een}{\end{enumerate}}

\renewcommand{\hat}{\widehat}
\renewcommand{\tilde}{\widetilde}
\renewcommand*{\epsilon}{\varepsilon}
\renewcommand*{\rho}{\varrho}

\makeatother

\usepackage{babel}

\institut{Institut f\"ur Analysis}

\preprintnumber{MATH-AN-02-2014}

\preprinttitle{On a Connection between the Maxwell System, the Extended Maxwell
System, the Dirac Operator and Gravito-Electromagnetism.}

\author{ Rainer Picard, Sascha Trostorff, \& Marcus Waurick}

\begin{document}

\makepreprinttitlepage

\author{ Rainer Picard, \\ Sascha Trostorff, \\ Marcus Waurick \\ Institut f\"ur Analysis, Fachrichtung Mathematik\\ Technische Universit\"at Dresden\\ Germany\\ rainer.picard@tu-dresden.de\\ sascha.trostorff@tu-dresden.de\\ marcus.waurick@tu-dresden.de } 

\title{On a Connection between the Maxwell System, the Extended Maxwell
System, the Dirac Operator and Gravito-Electromagnetism.}

\maketitle

\begin{abstract} \textbf{Abstract.} Maxwell's equation, Dirac's equation and the equation of gravito-electromagnetism are shown to be particular instances of the extended Maxwell system. The equations are discussed in the framework of the theory of evolutionary equations. Their formal relationship are systematically analyzed. Applications to coupled systems such as the Maxwell-Dirac system are also discussed.  \end{abstract}


\begin{class} 35Q61 \end{class}

\newpage


\setcounter{section}{-1}

\maketitle

\tableofcontents{}

\newpage

\section{Introduction}

We shall approach the different systems mentioned in the title as
instances of the typical linear case of a general space-time operator
equation
\begin{equation*}
\partial_{0}V+AU=f,
\end{equation*}
where $f$ comprises given data, $\partial_0$ denotes the derivative with respect to time,  $A$ is a -- usually -- purely spatial
operator and the quantities $U,V$ are linked by a so-called material
law
\begin{equation*}
V=\mathcal{M}U.
\end{equation*}
Then the closed operator sum in question is $\overline{\partial_{0}\mathcal{M}+A}$.
As a matter of ``philosophy'', for the way we like to think about
this problem class, the material law operator $\mathcal{M}$ is encoding
the complexity of the physical material, whereas $A$ is kept simple
and usually only contains spatial derivatives. If $\mathcal{M}$ commutes
with $\partial_{0}$ we speak of an autonomous system, a case which
we will focus on here as a matter of simplification (for a treatment of non-autonomous problems see \cite{Picard2013_nonauto, Waurick2013_nonauto}).

A prominent feature distinguishing general operator equations from
those describing dynamic processes is the specific role of time, which
is not just another space variable, but characterizes dynamic processes
via the property of causality. Requiring causality for the solution
operator $\overline{\partial_{0}\mathcal{M}+A}^{-1}$ results in very
specific types of material law operators $\mathcal{M}$, which are
causal and compatible with causality of $\overline{\partial_{0}\mathcal{M}+A}^{-1}$.
This leads to deeper insights into the structural properties of
mathematically viable models of physical phenomena. Causality suggests
to think of prescribed initial data, which can be achieved by extending
the solution theory canonically to temporal distributions with values
in a Hilbert space. In this perspective considering an initial value
problem, i.e. prescribing $V\left(0+\right)$, amounts to allowing
a source term $f$ of the form $\delta\otimes V_{0}$ defined by
\begin{equation*}
\delta\otimes V_{0}\left(\varphi\right)\coloneqq\left\langle V_{0}|\varphi\left(0\right)\right\rangle _{H}
\end{equation*}
for $\varphi$ in the space $\interior{C}\left(\mathbb{R},H\right)$ of continuous
$H$-valued functions with compact support. This source term encodes
indeed the classical initial condition $V\left(0+\right)=V_{0}$. 

Throughout, the scalar field of any Hilbert space $H$ under consideration
is the field of complex numbers, although we shall implicitly consider
this type of Hilbert space occasionally as a real Hilbert space in
so far as we consider positivity with respect to the corresponding
real inner product $\Re\left\langle \:\cdot\:|\:\cdot\:\right\rangle _{0}$
derived from the actual inner product $\left\langle \:\cdot\:|\:\cdot\:\right\rangle _{0}$
of $H$. Note that inner products, indeed any sesquilinear forms,
are assumed to be conjugate-linear in the first component and linear
in the second component. 

Taking Maxwell's equations as a starting example, we recall that these describe the evolution of the electro-magnetic
field $(E,H)$ in a $3$-dimensional open set $\Omega$. As Gauss's
law can be incorporated by a suitable choice of initial data, we think
of Maxwell's equations as being composed of two separate laws of which
the first is Faraday's law of induction (the Maxwell-Faraday equation),
which reads as 
\begin{equation*}
\partial_{0}B+\interior\curl E=0,
\end{equation*}
where $\interior\curl$ denotes the (distributional) $\curl$ operator
in $L^{2}(\Omega)^{3}$ with the electric boundary condition. $B$,
the magnetic induction, satisfies the constitutive equation
\begin{equation*}
B=\mu H,
\end{equation*}
where $\mu$ is the magnetic permeability. Faraday's law is complemented
by Ampere's law
\begin{equation*}
\partial_{0}D+J_{c}-\curl H=J_{0}
\end{equation*}
for $J_{0}$, $D$, $J_{c}$ being the external currents, the electric
displacement and the induced current, respectively. The latter two
quantities satisfy the two equations
\begin{align*}
D & =\epsilon E,\text{ and }\\
J_{c} & =\sigma E.
\end{align*}
The former is a constitutive equation involving the dielectricity
$\epsilon$ and the latter is Ohm's law with the conductivity $\sigma$. Plugging the constitutive
relations and Ohm's law into Faraday's law of induction and Ampere's
law and arranging them in a block operator matrix equation, we arrive
at 
\begin{equation*}
\left(\partial_{0}\left(\begin{array}{cc}
\epsilon & 0\\
0 & \mu
\end{array}\right)+\left(\begin{array}{cc}
\sigma & 0\\
0 & 0
\end{array}\right)+\left(\begin{array}{cc}
0 & -\curl\\
\interior\curl & 0
\end{array}\right)\right)\left(\begin{array}{c}
E\\
H
\end{array}\right)=\left(\begin{array}{c}
J_{0}\\
0
\end{array}\right).
\end{equation*}

The above form of Maxwell's equations is actually due to O. Heaviside,
which is therefore sometimes referred to as the Maxwell-Heaviside
system. Maxwell used a formalization based on quaternions, which were
discovered independently by O. Rodrigues (1840) and A. Hamilton (1843).
Following the formal implications of the quaternion calculus T. Ohmura,
\cite{Ohmura1956}, appears to have been the first to suggest explicitly
an extended Maxwell system, which, as a by-product, allowed to incorporate
charge preservation into the system of Maxwell's equations, without
causing overdeterminancy, compare also \cite{0305-4470-28-17-030-1}.
This system in turn coincides in its Euclidean (non-quaternionic)
form with a system proposed by Dirac, \cite{Dirac01091931}, for the
construction of electrodynamic potentials in the presence of magnetic
and electric charges. 

Although Dirac has developed his equations from a direct formal calculation
to find a first order factorization for the Klein-Gordon operator,
the Dirac equation is easily recognized as of quaternionic nature
(The Pauli matrices multiplied by the imaginary unit $\ii$ together
with the identity form a unitary quaternion basis). 

A variant of the previously discussed extensions of the Maxwell system
has been proposed in a quaternionic framework as a linearized version
of the gravitation equations by A. Singh, \cite{citeulike:3940266}.
For this approach the label gravito-electromagnetism or GEM has been
used, compare e.g. \cite{Ulrych2006} as a more recent reference.
The initiating idea, however, dates back to O. Heaviside, \cite{OH1893}.

An independent strand of development is rooted in the use of alternating
differential forms for the formulation of the Maxwell-Heaviside system.
The rigorous investigation of initial boundary value problems for
Maxwell's equations in the alternating differential form setting goes
back to H. Weyl, \cite{Ref730}, and K. O. Friedrichs, \cite{Ref165}.
We mention in passing that, as noticed by N. Weck, the differential
forms framework proved to be essential in showing a compact resolvent
result for the spatial derivative operator of Maxwell's equations
in the case of non-smooth boundaries, see \cite{Ref697}. Such types
of results were previously derived from Rellich's compact embedding
result, which, however, fails to hold even for simple non-smooth boundaries
such as for example a unit cube $\oi01^{3}$ with a smaller cube $\left[0,\epsilon\right]^{3}$,
$\epsilon\in\oi01$, removed. 

Returning to our purpose here, we note that by following the intrinsic
logic of the differential forms setting, a functional analytical framework
for initial boundary value problems for Maxwell's equations can be
embedded in an extended first order system on Riemannian manifolds,
\cite{0579.58030}. Considering the Euclidean free space case, it
turns out that this extension formally coincides with the Ohmura system, \cite{Ohmura1956}. 

That alternating tensors can also be used to formulate quantum dynamics
has been observed and proposed (almost simultaneously to the publication
of the celebrated Dirac system) by Iwanenko and Landau, \cite{zbMATH02579750}.
Unfortunately, the idea has not found the attention of the main stream
at the time, which eventually led to the independent rediscovery of
the possibility of reformulating quantum dynamics in 4-dimensional
alternating differential forms by E. K\"ahler, \cite{0127.31402}. Again,
in its Euclidean form, these equations turn out to coincide with the
above extended Maxwell system.

It is again the Clifford algebra setting in which the connection between
the Dirac equation and Maxwell's equations has first been established,
see e.g. \cite{0879.15023,Krav}. 

Despite of the possibility to discuss the interconnectedness of GEM,
the Dirac equation and Maxwell equations within the general framework
of alternating differential forms on Riemannian manifolds, as a matter
of accessibility and simplicity we shall here focus on the Euclidean
case in open domains of $\mathbb{R}^{3}$. Our main interest is to
establish well-posedness results for related initial boundary value
problems and to discuss the intimate connections between these systems
and their impact on some applications. 

In section \ref{sec:Maxwell's-System} we shall start by establishing
Maxwell's equations in the framework of a functional analytical approach
developed in \cite{Pi2009}. This approach hinges on the observation
that the time-derivative can be established as a normal operator.
In the present context, this observation allows for a transparent
discussion of the extended Maxwell system leading to new insights
into this construction, see subsection \ref{sub:The-Extended-Maxwell}.
In the sections \ref{sec:GEM} and \ref{sec:The-Dirac-System} we establish
the links to the GEM and the Dirac equation. The last section is
dedicated to a particular coupled system of Maxwell's equations for
an electromagnetic field, the GEM type system for the associated potentials
and the Dirac equation, which in the light of our observations turns
into a coupled system of three extended Maxwell systems.

\section{\label{sec:Maxwell's-System}Maxwell's System}

\subsection{Maxwell's equations as an Evolutionary Problem}

We recall from \cite{Pi2009} that the temporal differentiation
operator denoted by $\partial_{0}$ can be established as a normal
operator in the exponentially weighted complex $L^{2}$-type space $H_{\nu,0}(\mathbb{R},H)$ given by 
\[
H_{\nu,0}(\mathbb{R},H)\coloneqq \{f:\mathbb{R}\to H\,|\,(t\mapsto f(t)\exp(-\nu t))\in L^2(\mathbb{R};H)\}, 
\]
where $\nu\in \oi{0}{\infty}$. We endow this space with its natural inner product $\left\langle \,\cdot\,|\,\cdot\,\right\rangle _{\nu,0,0}$ defined by 
\begin{equation*}
\left(u,v\right)\mapsto\int_{\mathbb{R}}\left\langle u\left(t\right)|v\left(t\right)\right\rangle _{0}\,\exp\left(-2\nu t\right)\, \d t
\end{equation*}
where $\left\langle \,\cdot\,|\,\cdot\,\right\rangle _{0}$ denotes
the inner product of the underlying complex Hilbert space $H$. If
we consider $H_{\nu,0}\left(\mathbb{R},H\right)$ as a real Hilbert
space with inner product $\Re\left\langle \,\cdot\,|\,\cdot\,\right\rangle _{\nu,0,0}$
then it turns out that $\partial_{0}$ is strictly positive definite and hence, continuously invertible. More precisely, we have
\[
 \Re \langle \partial_0 u|u\rangle_{\rho,0,0}= \rho|u|_{\rho,0,0}^2.
\]
Moreover, we define the extrapolation space $H_{\rho,k}(\mathbb{R},H)$ for $k\in \mathbb{Z}$ associated with $\partial_0$ as the completion of $D(\partial_0^k)$ with respect to the inner product
\begin{equation*}
 (u,v) \mapsto \langle \partial_0^k u| \partial_0^k v\rangle_{\rho,0,0} \eqqcolon \langle u|v\rangle_{\rho,k,0}.
\end{equation*}
In generalization of our earlier discussed material relations we obtain Maxwell's equations in the
form 
\begin{equation}
\left(\partial_{0}M_{0}+M_{1}+\left(\begin{array}{cc}
0 & -\curl\\
\interior\curl & 0
\end{array}\right)\right)\left(\begin{array}{c}
E\\
H
\end{array}\right)=\left(\begin{array}{c}
-J\\
K
\end{array}\right) \label{eq:maxwell}.
\end{equation}
Here, $\interior\curl$ is defined as the closure of
\begin{align*}
 \curl|_{\interior{C}_\infty(\Omega)} : \interior{C}_\infty(\Omega)^3\subseteq L^2(\Omega)^3 &\to L^2(\Omega)^3 \\
              (\phi_1,\phi_2,\phi_3) &\mapsto \left(\begin{array}{ccc}
0 & -\partial_{3} & \partial_{2}\\
\partial_{3} & 0 & -\partial_{1}\\
-\partial_{2} & \partial_{1} & 0
\end{array}\right) \begin{pmatrix}
                    \phi_1 \\
                    \phi_2\\
                    \phi_3
                   \end{pmatrix},
\end{align*}
where $\Omega\subseteq \mathbb{R}^3$ is any open set\footnote{Note that we do not need to impose any additional geometric constraints on $\Omega$, since boundary trace results and compact embedding results are not required for a basic well-posedness result.} and $\interior{C}_\infty(\Omega)$ denotes the set of infinitely differentiable functions with compact support in $\Omega$. The operator $\curl$
is then defined as its adjoint $\interior\curl^{*}$. As a consequence, the operator
\begin{equation*}
\left(\begin{array}{cc}
0 & -\curl\\
\interior\curl & 0
\end{array}\right) 
\end{equation*}
is skew-selfadjoint. Moreover, by the definition of $\interior{\curl}$ an implicit boundary condition in (\ref{eq:maxwell}) is prescribed for the unknown $E$, namely the electric boundary condition $n\times E=0$, where $n$ denotes the outward unit normal vector field on $\partial \Omega$ (indeed, the condition $E\in D(\interior{\curl})$ is a generalization of the electric boundary condition, since we can define the operator $\interior{\curl}$ on domains, where no unit outward normal vector field exists). 

\begin{rem}
Although any other boundary condition for $\left(\begin{array}{cc}
0 & -\curl\\
\curl & 0
\end{array}\right)$ producing -- say -- an m-accretive operator (see \cite{Trostorff2013_maxmon_bd} for a characterisation of those boundary conditions), under the constraint on the coefficients
discussed here, yield well-posedness, we focus here on the standard,
skew-selfadjoint case, since our goal is to gain structural insights
in connection with the extended Maxwell system.  
\end{rem}

The material law in (\ref{eq:maxwell}) is given by 
\begin{equation}{\label{eq:mat_law}}
 \mathcal{M}=M_0+\partial_0^{-1}M_1, 
\end{equation}
where the operators $M_0$ and $M_1$  are arbitrary bounded linear operators on
$H\coloneqq L^{2}\left(\Omega\right)^{3}$ satisfying the two requirements 
\begin{itemize}
\item \label{(H1)--selfadjoint}(\textbf{H1}) $M_{0}$ selfadjoint,
\item \label{(H2)-posdef}(\textbf{H2}) $\nu M_{0}+\Re M_{1}\geq c_{0}$
for some $c_{0}\in\oi0\infty$ and all sufficiently large $\nu\in\oi0\infty$.
\end{itemize}
Under these assumptions the operator $\left(\partial_{0}M_{0}+M_{1}+\left(\begin{array}{cc}
0 & -\curl\\
\interior\curl & 0
\end{array}\right)\right)$ together with its adjoint are both strictly positive definite obviously
resulting in well-posedness of (\ref{eq:maxwell}), compare \cite{Pi2009}. Moreover, the solution operator is \emph{causal}, i.e., denoting \[S\coloneqq \overline{\left(\partial_{0}M_{0}+M_{1}+\left(\begin{array}{cc}
0 & -\curl\\
\interior\curl & 0
\end{array}\right)\right)}^{-1}\in L(H_{\rho,0}(\mathbb{R},H),H_{\rho,0}(\mathbb{R},H))\] we have the operator equation
\[
   \chi_{]-\infty,0[}(\mathrm{m_0})S\chi_{]-\infty,0[}(\mathrm{m_0})=\chi_{]-\infty,0[}(\mathrm{m_0})S,
\]
 where $\chi_{]-\infty,0[}(\mathrm{m_0})$ denotes the multiplication operator on $H_{\rho,0}(\mathbb{R},H)$ induced by the cut-off function $t\mapsto \chi_{]-\infty,0[}(t)$.

\begin{rem} $\,$
\begin{enumerate}[(a)]
 \item Indeed, in \cite{Pi2009} the well-posedness is proved for a broader class of possible material laws $\mathcal{M}$, which are given as analytic operator-valued functions of $\partial_0^{-1}$. However, for sake of simplicity we restrict ourselves to the ``affine'' type situation (\ref{eq:mat_law}).
\item It is worth noting, that although in restricted cases well-posedness
of Maxwell's equations is a classical result, it is still a new result
in the current generality, since here $M_{0}$, $M_{1}$ are operators
(and not just matrices or matrix-valued functions with $L^{\infty}\left(\Omega\right)$-entries).
So, in particular non-local coefficients are painlessly included.
Even if we consider only the matrix-valued case, the well-posedness
result obtained is surprisingly powerful. We have \textbf{not }required
$M_{0}$ to be strictly positive definite, but merely that (\textbf{H1}),
(\textbf{H2}) hold. Thus, we have without additional effort not only
included eddy current type problems, where
\begin{equation*}
M_{0}=\left(\begin{array}{cc}
0 & 0\\
0 & \mu
\end{array}\right),\: M_{1}=\left(\begin{array}{cc}
\sigma & 0\\
0 & 0
\end{array}\right),
\end{equation*}
with $\mu,\sigma$ strictly positive definite, but indeed mixed problems,
where $M_{0}$ is strictly positive definite in some part of $\Omega$
and we have an eddy current type situation in another part of $\Omega$.
A typical model situation for the latter case is the transmission
problem between a transformer core with a gap and its complement,
if within the transformer core an eddy current approximation is used,
whereas in the complement the Maxwell system with a strictly positive
definite $M_{0}$ is assumed. With the usual methods (semi-groups,
symmetric hyperbolic systems, spectral theory etc.) such a model situation
is not as easily accessible. 
\item It is also noteworthy that we do not assume that $M_{0},M_{1},$ are
block-diagonal, as in the classical anisotropic inhomogeneous media
case, so general bi-anisotropic and chiral media are covered by the
same description. 
\end{enumerate}
\end{rem}

\begin{rem}[See-saw regularity and initial-value problems]$\,$
\begin{enumerate}[(a)]
\item With the help of the notion of Sobolev lattices the equation
\begin{equation}\label{eq:max_see_saw}
  \left(\partial_{0}M_{0}+M_{1}+\left(\begin{array}{cc}
0 & -\curl\\
\interior\curl & 0
\end{array}\right)\right)U=F
\end{equation}
 may also be read line by line. The price one has to pay is that even if $F\in H_{\rho,k}(\mathbb{R},H)$ only the solution $U$ and $\left(\partial_{0}M_{0}+M_{1}+\left(\begin{array}{cc}
0 & -\curl\\
\interior\curl & 0
\end{array}\right)\right)U$ belong to  $H_{\rho,k}(\mathbb{R},H)$. But as we suppressed to write the closure for avoiding to clumsy a notation we observe that the constituents of the sum on the right-hand side only satisfy the following
\begin{multline*}
   \partial_{0}M_{0}U\in H_{\rho,k-1}(\mathbb{R},H),\ M_{1}U\in H_{\rho,k}(\mathbb{R},H),\ \partial_{0}M_{0}U\in H_{\rho,k-1}(\mathbb{R},H),\\
  \left(\begin{array}{cc}
0 & -\curl\\
\interior\curl & 0
\end{array}\right)\begin{pmatrix}
                     U_+ \\ U_-
                  \end{pmatrix} \in H_{\rho,k}(\mathbb{R},H_{-1}(|\interior\curl|+\i)\oplus H_{-1}(|{\curl}|+\i)),  
\end{multline*}
where we used the notion of Sobolev lattices, see e.g.~\cite{PDE_DeGruyter} (and in particular \cite[p. 50]{PDE_DeGruyter}). Later, we shall need that $U$ in fact attains values in the domain of the Maxwell operator if interpreted suitably. This can be made precise as follows. Rearranging the terms in \eqref{eq:max_see_saw}, we get that
\[
   \left(\begin{array}{cc}
0 & -\curl\\
\interior\curl & 0
\end{array}\right)U=F-\partial_{0}M_{0}U+M_{1}U\in H_{\rho,k-1}(\mathbb{R},H).
\]
Hence, from $U\in H_{\rho,k}(\mathbb{R},H)\subseteqq H_{\rho,k-1}(\mathbb{R},H)$ we infer that
\begin{multline*}
   U=\left(1+\left(\begin{array}{cc}
0 & -\curl\\
\interior\curl & 0
\end{array}\right)\right)^{-1}\left(F-\partial_{0}M_{0}U+M_{1}U+U\right)\\ \in H_{\rho,k-1}(\mathbb{R},H_{1}(|\interior\curl|+\i)\oplus H_{1}(|{\curl}|+\i)).
\end{multline*}
In particular, for $F$ belonging to the inductive limit $H_{\nu,-\infty}\left(\mathbb{R},H\right)\coloneqq\bigcup_{k\in\mathbb{Z}}H_{\nu,k}\left(\mathbb{R},H\right)$, we get that $U\in H_{\rho,-\infty}(\mathbb{R},H_{1}(|\interior\curl|+\i)\oplus H_{1}(|{\curl}|+\i))$. 

\item  If we want to include initial values in Maxwell's equations 
\begin{equation} \label{eq:maxwell_abstract}
\left(\partial_{0}M_{0}+M_{1}+\left(\begin{array}{cc}
0 & -\curl\\
\interior\curl & 0
\end{array}\right)\right)U=F, 
\end{equation} we may assume that the given right hand side $F$ vanishes on $\rci{-\infty}{0}$ and we additionally assume that\footnote{Note that it only makes sense to prescribe an initial value for the part of $U$, which gets differentiated with respect to time.} $M_0 U(0+)=M_0 U_0$ for a given element $U_0\in H$. There are two ways to incorporate this initial value problem into our abstract framework. 
\begin{itemize}
 \item If $U_0\in D\left(\left(\begin{array}{cc}
0 & -\curl\\
\interior\curl & 0
\end{array}\right)\right)$ then we could consider the problem \begin{multline}
\left(\partial_{0}M_{0}+M_{1}+\left(\begin{array}{cc}
0 & -\curl\\
\interior\curl & 0
\end{array}\right)\right)\left(U-\chi_{_{\oi0\infty}}\otimes U_{0}\right) \\ =F-\chi_{_{\oi0\infty}}\otimes\left(\begin{array}{cc}
0 & -\curl\\
\interior\curl & 0
\end{array}\right)U_{0}-\chi_{\oi{0}{\infty}}\otimes M_1U_0\label{eq:initial1},
\end{multline}
which is well-posed according to our solution theory. Clearly, $U$ satisfies the differential equation (\ref{eq:maxwell_abstract}) for positive times. Moreover, due to the causality of the solution operator, we get that $U-\chi_{_{\oi0\infty}}\otimes U_{0}$  is supported on the positive real axis. Furthermore, from \eqref{eq:initial1} we read off that \[M_0(U-\chi_{\oi{0}{\infty}}\otimes U_0)\in H_{\nu,1}(\mathbb{R};H_{-1}(A+1)),\] where $H_{-1}(A+1)$ denotes the first extrapolation space associated with the operator $A+1$ (see \cite[Section 2.1]{PDE_DeGruyter}). Thus, by a version of the Sobolev embedding theorem (see \cite[Lemma 3.1.59]{PDE_DeGruyter}), we obtain the continuity of $M_0(U-\chi_{_{\oi0\infty}}\otimes U_{0})$ yielding
\[
 M_0(U-\chi_{_{\oi0\infty}}\otimes U_{0})(0+)=M_0(U-\chi_{_{\oi0\infty}}\otimes U_{0})(0-)=0, 
\]
 which implies $M_0 U(0+)=M_0 U_0$.
\item  Alternatively, for including initial data we note that we can extend
the above solution theory to $H_{\nu,-\infty}\left(\mathbb{R},H\right)$,
and that initial data can then be modeled as Dirac-$\delta$-sources.
Indeed, noting that $\partial_{0}\chi_{_{\oi0\infty}}=\delta$
we get 
\begin{equation}
\left(\partial_{0}M_{0}+M_{1}+\left(\begin{array}{cc}
0 & -\curl\\
\interior\curl & 0
\end{array}\right)\right)U=F+\delta\otimes M_0 U_{0}\label{eq:initial2}
\end{equation}
as an equation in $H_{\nu,-\infty}\left(\mathbb{R},H\right)$ modelling
the prescription of initial data. 
\end{itemize}
\end{enumerate}
\end{rem}

In view of our perspective on the extended Maxwell system, which we
wish to approach next, we consider Maxwell's equations in a more simple
situation, where actually $M_{1}=0$. As a consequence (\textbf{H1}),
(\textbf{H2}) imply that $M_{0}$ is selfadjoint and strictly positive
definite. In this situation, classical methods can of course easily
be employed. Indeed, we have an abundance of structure. Obviously,
equation (\ref{eq:maxwell}) can in this case be re-written as
\begin{equation}
\left(\partial_{0}+\sqrt{M_{0}^{-1}}\left(\begin{array}{cc}
0 & -\curl\\
\interior\curl & 0
\end{array}\right)\sqrt{M_{0}^{-1}}\right)V=\sqrt{M_{0}^{-1}}\left(\begin{array}{c}
-J\\
K
\end{array}\right)\eqqcolon F,\label{eq:maxwell-m1-0}
\end{equation}
where $V=\sqrt{M_{0}}\left(\begin{array}{c}
E\\
H
\end{array}\right)$. Here now $\sqrt{M_{0}^{-1}}\left(\begin{array}{cc}
0 & -\curl\\
\interior\curl & 0
\end{array}\right)\sqrt{M_{0}^{-1}}$ is still skew-selfadjoint and commuting with the normal operator
$\partial_{0}$ and so we can consider (\ref{eq:maxwell-m1-0}) in
the convenient framework of commuting selfadjoint operators, for which
we can approach solvability simply by noting that the operator occurring
on the left-hand side is actually given by a function evaluation $f\left(\Im\partial_{0},\frac{1}{\ii}\sqrt{M_{0}^{-1}}\left(\begin{array}{cc}
0 & -\curl\\
\interior\curl & 0
\end{array}\right)\sqrt{M_{0}^{-1}}\right)$, where%
\footnote{Invertibility can be read off, since the solution operator is simply
$g\left(\Im\partial_{0},\frac{1}{\ii}\sqrt{M_{0}^{-1}}\left(\begin{array}{cc}
0 & -\curl\\
\interior\curl & 0
\end{array}\right)\sqrt{M_{0}^{-1}}\right)$ with 
\begin{align*}
g:\mathbb{R}\times\mathbb{R} & \to\mathbb{C},\\
\left(x,y\right) & \mapsto\frac{1}{\ii\left(x+y\right)+\nu}.
\end{align*}
This solution can -- for suitable data -- be represented by utilizing
the fundamental solution associated with (\ref{eq:maxwell-m1-0}),
which is given by 
\begin{equation*}
t\mapsto\chi_{_{\lci0\infty}}\left(t\right)\exp\left(\ii t\frac{1}{\ii}\sqrt{M_{0}^{-1}}\left(\begin{array}{cc}
0 & -\curl\\
\interior\curl & 0
\end{array}\right)\sqrt{M_{0}^{-1}}\right)=\chi_{_{\lci0\infty}}\left(t\right)\exp\left(t\sqrt{M_{0}^{-1}}\left(\begin{array}{cc}
0 & -\curl\\
\interior\curl & 0
\end{array}\right)\sqrt{M_{0}^{-1}}\right),
\end{equation*}
 via a convolution integral
\begin{align*}
V\left(t\right) & =\int_{\mathbb{R}}\chi_{_{\lci0\infty}}\left(t-s\right)\exp\left(\left(t-s\right)\sqrt{M_{0}^{-1}}\left(\begin{array}{cc}
0 & -\curl\\
\interior\curl & 0
\end{array}\right)\sqrt{M_{0}^{-1}}\right)F\left(t\right)\: dt\\
 & =\int_{-\infty}^{t}\exp\left(\left(t-s\right)\sqrt{M_{0}^{-1}}\left(\begin{array}{cc}
0 & -\curl\\
\interior\curl & 0
\end{array}\right)\sqrt{M_{0}^{-1}}\right)F\left(t\right)\: dt\\
 & =\exp\left(-s\sqrt{M_{0}^{-1}}\left(\begin{array}{cc}
0 & -\curl\\
\interior\curl & 0
\end{array}\right)\sqrt{M_{0}^{-1}}\right)\int_{-\infty}^{t}\exp\left(t\sqrt{M_{0}^{-1}}\left(\begin{array}{cc}
0 & -\curl\\
\interior\curl & 0
\end{array}\right)\sqrt{M_{0}^{-1}}\right)F\left(t\right)\: dt.
\end{align*}
} 
\begin{align*}
f:\mathbb{R}\times\mathbb{R} & \to\mathbb{C},\\
\left(x,y\right) & \mapsto\ii\left(x+y\right)+\nu.
\end{align*}
We shall take this observation, which obviously extends to more general
conservative media, as a starting point to discuss the extended Maxwell
system.

\subsection{\label{sub:The-Extended-Maxwell}The Extended Maxwell System}

The extended Maxwell system amounts to a specific coupling of acoustic
equations with Maxwell's equations, \cite{0541.35049}. To fix some notation, we introduce the following differential operators.
\begin{defn}
 Let $\Omega \subseteq \mathbb{R}^n$ be an arbitrary open subset. We define the operator $\interior{\grad}$ as the closure of
\begin{align*}
 \grad|_{\interior{C}_\infty(\Omega)}:\interior{C}_\infty(\Omega)\subseteq L^2(\Omega) & \to L^2(\Omega)^3 \\
          \phi &\mapsto (\partial_1 \phi, \partial_2\phi, \partial_3\phi).
\end{align*}
Likewise, we define $\interior{\dive}$ as the closure of
\begin{align*}
 \dive|_{\interior{C}_\infty(\Omega)}:\interior{C}_\infty(\Omega)^3\subseteq L^2(\Omega)^3 & \to L^2(\Omega) \\
          (\psi_1,\psi_2,\psi_3) &\mapsto \sum_{i=1}^3\partial_i\psi_i.
\end{align*}
Moreover, we set $\grad\coloneqq -(\interior{\dive})^\ast$ and $\dive\coloneqq -(\interior{\grad})^\ast$ and obtain $\interior{\grad}\subseteq \grad$ and $\interior{\dive}\subseteq \dive$. Furthermore we introduce the following block-operator matrices:
\begin{equation*}
 A_{\mathrm{Dac}}\coloneqq \begin{pmatrix}
                            0 & \dive &0 &0 \\
                            \interior{\grad} &0 &0 &0 \\
                            0&0&0&0\\
                            0&0&0&0
                           \end{pmatrix},
\end{equation*}
arising in the study of acoustics with Dirichlet boundary conditions, 
\begin{equation*}
 A_{\mathrm{Nac}}\coloneqq \begin{pmatrix}
                            0 & 0 &0 &0 \\
                            0& 0 &0 &0 \\
                            0&0&0 & \grad\\
                            0&0&\interior{\dive}& 0 
                           \end{pmatrix},
\end{equation*}
denoting the acoustic operator with Neumann boundary conditions, and the Maxwell operator
\begin{equation*}
 A_{\mathrm{Max}}\coloneqq \begin{pmatrix}
                            0 & 0 &0 &0 \\
                            0& 0 &-\curl &0 \\
                            0&\interior{\curl}&0 & 0\\
                            0&0&0& 0 
                           \end{pmatrix}.
\end{equation*}
Finally, we set $A_{\mathrm{ac}}\coloneqq A_{\mathrm{Dac}}+A_{\mathrm{Nac}}$.
\end{defn}

\begin{rem}\label{rem:annihilate}
 We note that the operators $A_{\mathrm{Dac}},A_{\mathrm{Nac}}$ and $A_{\mathrm{Max}}$ are all skew-selfadjoint. Moreover, since $A_{\mathrm{Dac}}$ and $A_{\mathrm{Nac}}$ commute, the operator $A_{\mathrm{ac}}$ is skew-selfadjoint as well. Finally, using the well-known relations $\curl \grad=0$ and $\dive \curl =0$, we derive
\begin{align*}
 A_{\Max}A_{\mathrm{ac}}&=0 \\
 A_{\ac}A_{\Max}&=0.
\end{align*}
\end{rem}

With these skew-selfadjoint operators the extended Maxwell system is
given by

\begin{align}{\label{eq:maxwell_ext}}
\left(\partial_{0}+\sqrt{\mathcal{E}^{-1}}A_{\mathrm{Max}}\sqrt{\mathcal{E}^{-1}}+\sqrt{\mathcal{E}}A_{\mathrm{ac}}\sqrt{\mathcal{E}}\right)V & =\tilde{F},
\end{align}
where $\mathcal{E}$ is a bounded, selfadjoint and strictly positive definite operator on $H\coloneqq L^2(\Omega)\oplus L^2(\Omega)^3\oplus L^2(\Omega)^3 \oplus L^2(\Omega)$ (taking over the role of $M_{0}$ in the previous section).
Setting $A\coloneqq \sqrt{\mathcal{E}^{-1}}A_{\mathrm{Max}}\sqrt{\mathcal{E}^{-1}}+\sqrt{\mathcal{E}}A_{\mathrm{ac}}\sqrt{\mathcal{E}}$, which is a skew-selfadjoint operator on $H$ according to Remark \ref{rem:annihilate}, we arrive at the canonical form 
\begin{equation*}
 (\partial_0+A)V=\tilde{F},
\end{equation*}
 and thus, our solution theory applies. Moreover, we have that \[\left(\Im\partial_{0},\:\frac{1}{\ii}\sqrt{\mathcal{E}^{-1}}A_{\mathrm{Max}}\sqrt{\mathcal{E}^{-1}},\:\frac{1}{\ii}\sqrt{\mathcal{E}}A_{\mathrm{ac}}\sqrt{\mathcal{E}}\right)\]
is a family of commuting, selfadjoint operators and thus,  in the framework of the associated function calculus, we have that
the operator on the left-hand side of (\ref{eq:maxwell_ext}) is now a function of this family given by 
\begin{equation*}
\left(x,y,z\right)\mapsto\ii\left(x+y+z\right)+\nu,
\end{equation*}
 whose invertibility is easily obtained. 
 
%
%

\begin{rem}
It is useful to consider this extended Maxwell system e.g. for
numerical purposes, see \cite{TaskinenV07,Weggler1-2012}, or low
frequency limits \cite{1119.35098,1179.35322,0541.35049}, since the
spatial operator is now a differential operator of elliptic type,
yielding a ``small'' (e.g. a finite dimensional) null space.
\end{rem}

The next theorem shows the interconnection between the extended and the original Maxwell system.

\begin{thm}\label{thm:connection_max_maxext}
 Let $V$ be a solution of (\ref{eq:maxwell_ext}) for a right-hand side $\tilde{F}\in H_{\rho,-\infty}(\mathbb{R};H)$. Then $V$ satisfies
\begin{equation}\label{eq:maxwell_null}
 (\partial_0+\sqrt{\mathcal{E}^{-1}}A_\Max \sqrt{\mathcal{E}^{-1}})V=F,
\end{equation}
 where $F\coloneqq \partial_0(\partial_0 +\sqrt{\mathcal{E}}A_\ac \sqrt{\mathcal{E}})^{-1} \tilde{F}$. Conversely, if $V$ satisfies (\ref{eq:maxwell_null}) for a right-hand side $F\in H_{\rho,-\infty}(\mathbb{R};H)\cap D(A_\ac \sqrt{\mathcal{E}})$, then $V$ solves (\ref{eq:maxwell_ext}) for $\tilde{F}\coloneqq (1+\partial_0^{-1}\sqrt{\mathcal{E}}A_\ac \sqrt{\mathcal{E}})F$.\\
Finally, if $V$ satisfies (\ref{eq:maxwell_null}) and $F$ and $\mathcal{E}$ are of the block form
\begin{equation*}
 F=\begin{pmatrix}
    0\\
    F_1\\
    F_2\\
    0
   \end{pmatrix}\mbox{ and }\mathcal{E}= \begin{pmatrix}
                             \mathcal{E}_{00} & 0&0&0\\
                             0& \mathcal{E}_{11}&\mathcal{E}_{12}&0\\
                             0 & \mathcal{E}_{21}& \mathcal{E}_{22}&0\\
                             0&0&0&\mathcal{E}_{33}
                             \end{pmatrix}
\end{equation*} then $V=(0,E,H,0)$, where $(E,H)$ is a solution (\ref{eq:maxwell-m1-0}) for the right-hand side $(F_1,F_2)$ and $M_0\coloneqq \begin{pmatrix}
                                                                                                                                                         \mathcal{E}_{11}&\mathcal{E}_{12}\\
                                                                                                                                                          \mathcal{E}_{21}& \mathcal{E}_{22}
                                                                                                                                                        \end{pmatrix}$.
\end{thm}

\begin{proof}
According to Remark \ref{rem:annihilate} we have that
\begin{equation*}
(\partial_0+\sqrt{\mathcal{E}^{-1}}A_\Max\sqrt{\E^{-1}}) (\partial_0+\sqrt{\E}A_\ac\sqrt{\E})=\partial_0(\partial_0+\sqrt{E^{-1}}A_\Max\sqrt{\E^{-1}}+\sqrt{\E}A_\ac\sqrt{\E})
\end{equation*}
and thus,
\begin{equation*}
 \partial_0^{-1}(\partial_0+\sqrt{\E^{-1}}A_\Max\sqrt{\E^{-1}}+\sqrt{\E}A_\ac\sqrt{\E})^{-1}=(\partial_0+\sqrt{\mathcal{E}^{-1}}A_\Max\sqrt{\E^{-1}})^{-1}(\partial_0+\sqrt{\E}A_\ac\sqrt{\E})^{-1}.
\end{equation*}
Let $V$ satisfy (\ref{eq:maxwell_ext}) for a right-hand side $\tilde{F}$ and define $F\coloneqq \partial_0(\partial_0 +\sqrt{\mathcal{E}}A_\ac \sqrt{\mathcal{E}})^{-1} \tilde{F}$ as in the assertion. Then 
\begin{align*}
\partial_0^{-1}V&=\partial_0^{-1}(\partial_0+\sqrt{\E^{-1}}A_\Max\sqrt{\E^{-1}}+\sqrt{\E}A_\ac\sqrt{\E})^{-1}\tilde{F}\\
&=  (\partial_0+\sqrt{\mathcal{E}^{-1}}A_\Max\sqrt{\E^{-1}})^{-1}\partial_0^{-1}F,
 \end{align*}
showing the first assertion. If $V$ solves (\ref{eq:maxwell_null}) for a right-hand side in $F\in D(A_\ac \sqrt{\E})$ then we compute
\begin{align*}
 & (\partial_0+\sqrt{\E}A_\ac\sqrt{\E})^{-1}(\partial_0+\sqrt{\E^{-1}}A_\Max\sqrt{\E^{-1}}+\sqrt{\E}A_\ac\sqrt{\E})\\
&=(\partial_0+\sqrt{\E}A_\ac\sqrt{\E})^{-1}(\partial_0+\sqrt{\mathcal{E}^{-1}}A_\Max\sqrt{\E^{-1}})^{-1}F\\
&= \partial_0^{-1}(\partial_0+\sqrt{\E^{-1}}A_\Max\sqrt{\E^{-1}}+\sqrt{\E}A_\ac\sqrt{\E})^{-1}F,
\end{align*}
which yields
\begin{equation*}
 (\partial_0+\sqrt{\E^{-1}}A_\Max\sqrt{\E^{-1}}+\sqrt{\E}A_\ac\sqrt{\E})V=\partial_0^{-1}(\partial_0+\sqrt{\E}A_\ac\sqrt{\E})F,
\end{equation*}
showing the second assertion. Assume now that $F$ and $\mathcal{E}$ are of the  block form stated above and let $V=(V_0,V_1,V_2,V_3)$ be the solution of (\ref{eq:maxwell_null}). We have that
\begin{equation*}
 \sqrt{\E^{-1}}A_\Max\sqrt{\E^{-1}}=\begin{pmatrix}
                                    0& \begin{pmatrix}
                                       0 & 0 \end{pmatrix} &0 \\
                                    \begin{pmatrix}
                                      0\\
                                      0
                                    \end{pmatrix} & \sqrt{M_0^{-1}} \begin{pmatrix}
                                                                      0& -\curl \\
                                                                     \interior{\curl} &0 
                                                                     \end{pmatrix} \sqrt{M_0^{-1}} & \begin{pmatrix}
                                      0\\
                                      0
                                    \end{pmatrix}\\
                                  0& \begin{pmatrix}
                                       0 & 0 \end{pmatrix} &0
                                   \end{pmatrix}
\end{equation*}
and hence, we obtain $\partial_0 V_0 =\partial_0 V_3=0$, which gives $V_0=V_3=0$ as well as
\begin{equation*}
  \left(\partial_0  +\sqrt{M_0^{-1}} \begin{pmatrix}
                                                                      0& -\curl \\
                                                                     \interior{\curl} &0 
                                                                     \end{pmatrix} \sqrt{M_0^{-1}}\right)\begin{pmatrix}
             V_1 \\
             V_2
            \end{pmatrix}= \begin{pmatrix}
                              F_1\\
                             F_2
                            \end{pmatrix}, 
\end{equation*}
which shows the assertion.
 
\end{proof}

\begin{rem}\label{rem:extended_maxwell}$\,$
 \begin{enumerate}[(a)]
  \item If $\E$ has the block structure
       \begin{equation*}
        \left(\begin{array}{cccc}
\mathcal{E}_{00} & \mathcal{E}_{01} & 0 & 0\\
\mathcal{E}_{10} & \mathcal{E}_{11} & 0 & 0\\
0 & 0 & \mathcal{E}_{22} & \mathcal{E}_{23}\\
0 & 0 & \mathcal{E}_{32} & \mathcal{E}_{33}
\end{array}\right)
       \end{equation*}
it follows that not only $\sqrt{\E^{-1}}A_\Max\sqrt{\E^{-1}}$ and $\sqrt{\E}A_\ac\sqrt{\E}$ annihilate each other but also that $\sqrt{\E}A_\Dac\sqrt{\E}$ and $\sqrt{\E}A_\Nac\sqrt{\E}$ annihilate each other. Hence, $(\partial_0,\sqrt{\E^{-1}}A_\Max\sqrt{\E^{-1}}+1,\sqrt{\E}A_\Dac\sqrt{\E}+1,\sqrt{\E}A_\Nac\sqrt{\E}+1)$ is a family of commuting, continuously invertible operators, allowing us to study the problem
\begin{equation*}
 (\partial_0+\sqrt{\E^{-1}}A_\Max\sqrt{\E^{-1}}+\sqrt{\E}A_\Dac\sqrt{\E}+\sqrt{\E}A_\Nac\sqrt{\E})U=\tilde{F}
\end{equation*}
 in the Sobolev-lattice $H_{\rho,-\infty}(\partial_0,\sqrt{\E^{-1}}A_\Max\sqrt{\E^{-1}}+1,\sqrt{\E}A_\Dac\sqrt{\E}+1,\sqrt{\E}A_\Nac\sqrt{\E}+1)$, see \cite{PDE_DeGruyter}.
\item For general $\E$ we have
\begin{align*}
 &(\partial_0+\sqrt{\E^{-1}}A_\Max\sqrt{\E^{-1}}+\sqrt{\E}A_\ac\sqrt{\E})(\partial_0-\sqrt{\E^{-1}}A_\Max\sqrt{\E^{-1}}+\sqrt{\E}A_\ac\sqrt{\E})=\\
&=\partial_0^2 -\left(\left(\sqrt{\E^{-1}}A_\Max\sqrt{\E^{-1}}\right)^2+\left(\sqrt{\E}A_\ac\sqrt{\E}\right)^2\right).
\end{align*}
In the special case, when $\E=1$ this is nothing but the wave operator
\begin{equation}\label{eq:Dirac-factors}
 \partial_0^2-\left(A_\Max^2+A_\ac^2\right)=\partial_0^2-\Delta\eqqcolon \square,
\end{equation}
with unusually boundary conditions. This leads in the time-harmonic case to a weakly singular Green's
tensor, which explains the mentioned interest for numerical methods.
\item It may be preferable to expose the Hamiltonian structure of the extended
Maxwell operator in the isotropic homogeneous case by applying the
unitary selfadjoint permutation block matrix $\left(\begin{array}{cccc}
0 & 0 & 0 & 1\\
0 & 1 & 0 & 0\\
0 & 0 & 1 & 0\\
1 & 0 & 0 & 0
\end{array}\right)$. Indeed, applying this transformation results in
\begin{equation}
 \left(\begin{array}{cccc}
0 & 0 & 0 & 1\\
0 & 1 & 0 & 0\\
0 & 0 & 1 & 0\\
1 & 0 & 0 & 0
\end{array}\right)\left(A_{\mathrm{Max}}+A_\ac\right)\left(\begin{array}{cccc}
0 & 0 & 0 & 1\\
0 & 1 & 0 & 0\\
0 & 0 & 1 & 0\\
1 & 0 & 0 & 0
\end{array}\right)=\left(\begin{array}{cc}
\left(\begin{array}{cc}
0 & 0\\
0 & 0
\end{array}\right) & \left(\begin{array}{cc}
\interior\dive & 0\\
-\curl & \interior\grad
\end{array}\right)\\
\left(\begin{array}{cc}
\grad & \interior\curl\\
0 & \dive
\end{array}\right) & \left(\begin{array}{cc}
0 & 0\\
0 & 0
\end{array}\right)
\end{array}\right),\label{eq:ham-max1}
\end{equation}
which is of the Hamiltonian form $\begin{pmatrix}
                                       0 & -C^\ast \\
                                       C & 0
                                      \end{pmatrix}$ with $C=\begin{pmatrix}
                                                                \grad & \interior{\curl} \\
                                                                   0 & \dive
                                                             \end{pmatrix}$.
To obtain a Hamiltonian structure for the operator $\sqrt{\E^{-1}}A_{\mathrm{Max}}\sqrt{\E^{-1}}+\sqrt{\E}A_\ac\sqrt{\E}$ we additionally have to assume that $\E$ is compatible with the permutation matrix in the sense that
\begin{equation*}
 \left(\begin{array}{cccc}
0 & 0 & 0 & 1\\
0 & 1 & 0 & 0\\
0 & 0 & 1 & 0\\
1 & 0 & 0 & 0
\end{array}\right)\E\left(\begin{array}{cccc}
0 & 0 & 0 & 1\\
0 & 1 & 0 & 0\\
0 & 0 & 1 & 0\\
1 & 0 & 0 & 0
\end{array}\right)=\begin{pmatrix}
                    \mathcal{G}_{00} & \begin{pmatrix}
                                          0 & 0 \\
                                          0 & 0
                                       \end{pmatrix}\\
                      \begin{pmatrix}
                                          0 & 0 \\
                                          0 & 0
                                       \end{pmatrix} & \mathcal{G}_{11}
                   \end{pmatrix}
\end{equation*}
for suitable operators $\mathcal{G}_{00},\mathcal{G}_{11}$.
\item As mentioned in the introduction
the extended Maxwell system is actually rooted in the differential
form calculus. In differential forms and for simplicity in the isotropic,
homogeneous case the extended Maxwell system takes on the simple form
\begin{equation}
\partial_{0}+\left(\interior d\wedge\right)-\left(\interior d\wedge\right)^{*},\label{eq:e-max}
\end{equation}
where $\interior d\wedge$ denotes the closure of the exterior derivative
initially defined on smooth differential forms with compact support
in the Riemannian manifold $M$. The adjoint $\left(\interior d\wedge\right)^{*}$
is taken with respect to the inner product of the direct sum $L^{0,2}\left(M\right)\oplus L^{1,2}\left(M\right)\oplus L^{2,2}\left(M\right)\oplus L^{3,2}\left(M\right)$,
where $L^{q,2}\left(M\right)$ denotes the $L^{2}$-space of alternating
differential forms of order $q\in\{0,1,2,3\}$. The block matrix structure
induced by the direct sum structure of the underlying space leads
to
\begin{equation*}
\left(\begin{array}{cccc}
0 & -\left(\interior d\wedge\right)^{*} & 0 & 0\\
\interior d\wedge & 0 & -\left(\interior d\wedge\right)^{*} & 0\\
0 & \interior d\wedge & 0 & -\left(\interior d\wedge\right)^{*}\\
0 & 0 & \interior d\wedge & 0
\end{array}\right)
\end{equation*}
corresponding to $A_{\mathrm{Max}}+A_\ac$
or unitarily equivalent on $L^{3,2}\left(M\right)\oplus L^{1,2}\left(M\right)\oplus L^{2,2}\left(M\right)\oplus L^{0,2}\left(M\right)$
\begin{equation*}
\left(\begin{array}{cccc}
0 & 0 & \interior d\wedge & 0\\
0 & 0 & -\left(\interior d\wedge\right)^{*} & \interior d\wedge\\
-\left(\interior d\wedge\right)^{*} & \interior d\wedge & 0 & 0\\
0 & -\left(\interior d\wedge\right)^{*} & 0 & 0
\end{array}\right)
\end{equation*}
corresponding to (\ref{eq:ham-max1}). The form of (\ref{eq:e-max})
obviously also extends to different dimensions, see \cite{0552.58033,Picard1984}.
For example, the 1-dimensional extended Maxwell operator is in Cartesian
block matrix form simply 
\begin{equation*}
\left(\begin{array}{cc}
0 & \partial_{1}\\
\partial_{1} & 0
\end{array}\right).
\end{equation*}
 \end{enumerate}

\end{rem}

\section{\label{sec:GEM}Gravito-Electromagnetism (GEM)}

The extended Maxwell system appear to be the linearization of the equations of gravitation (see \cite{citeulike:3940266}, where a formulation within the framework of quaternions is provided).
More recently,  a reduced version of the extended Maxwell system  was used in the context of modern \textbf{gravito-electromagnetism}
(\textbf{GEM}), \cite{Ulrych2006}. The equation reads as 
\begin{equation}\label{eq:gem}
 \left(\partial_0+\sqrt{\C^{-1}}\begin{pmatrix}
                            0 & 0 &0 \\
                            0 & 0 & -\curl \\
                            0& \interior{\curl} &0
                           \end{pmatrix}\sqrt{\C^{-1}} +\sqrt{\C} \begin{pmatrix}
                                                                  0 & \dive & 0 \\
                                                                  \interior{\grad} & 0 &0 \\
                                                                   0 & 0 & 0
                                                                  \end{pmatrix} \sqrt{\C}\right)\begin{pmatrix}
                                                                                                 C\\
                                                                                                 E\\
                                                                                                 H
                                                                                               \end{pmatrix} = F,
\end{equation}
where $\C$ is assumed to be a selfadjoint strictly positive continuous operator on $H\coloneqq L^2(\Omega)\oplus L^2(\Omega)^3\oplus L^2(\Omega)^3$ and $F\in H_{\rho,0}(\mathbb{R};H)$ is a given source term. As in Remark \ref{rem:annihilate} we obtain that the spatial operators
\begin{equation*}
 \sqrt{\C^{-1}}\begin{pmatrix}
                            0 & 0 &0 \\
                            0 & 0 & -\curl \\
                            0& \interior{\curl} &0
                           \end{pmatrix}\sqrt{\C^{-1}} \mbox{ and } \sqrt{\C} \begin{pmatrix}
                                                                   0 & \dive & 0 \\
                                                                   \interior{\grad} & 0 &0 \\
                                                                    0 & 0 & 0
                                                                   \end{pmatrix} \sqrt{\C}
\end{equation*}
annihilate each other, yielding that 
\begin{equation*}
 A=\sqrt{\C^{-1}}\begin{pmatrix}
                            0 & 0 &0 \\
                            0 & 0 & -\curl \\
                            0& \interior{\curl} &0
                           \end{pmatrix}\sqrt{\C^{-1}} + \sqrt{\C} \begin{pmatrix}
                                                                   0 & \dive & 0 \\
                                                                   \interior{\grad} & 0 &0 \\
                                                                    0 & 0 & 0
                                                                   \end{pmatrix} \sqrt{\C}
\end{equation*}
defines a skew-selfadjoint operator on $H$. Obviously, if $(C,E,H)$ solves (\ref{eq:gem}) for some right-hand side $F=(F_0,F_1,F_2)$ then 
\begin{equation*}
 (\partial_0+\sqrt{\E^{-1}}A_\Max\sqrt{\E^{-1}}+\sqrt{\E}A_\Dac\sqrt{\E})V=\tilde{F},
\end{equation*}
 where $V=(C,E,H,0),\tilde{F}=(F_0,F_1,F_2,0)$ and $\E=\begin{pmatrix}
                                                        \C &0\\
                                                        0 & 1
                                                       \end{pmatrix}$ and vice versa. The connection of the latter system with the extended Maxwell system is stated in the following theorem.

\begin{thm}
Let $\C$ be of the form $\begin{pmatrix}
                          \C_{00} & \C_{01} &0\\
                          \C_{10} & \C_{11} &0 \\
                           0&0&\C_{22}
                         \end{pmatrix}$. Let $K\in L(L^2(\Omega))$ be selfadjoint and $S\in L(L^2(\Omega),L^2(\Omega)^3)$ such that $K-S^\ast \C^{-1} S$ is strictly positive definite.  
\begin{equation*}
 \E\coloneqq \begin{pmatrix}
              \C & \begin{pmatrix}
                    0\\
                    0\\
                    S
             \end{pmatrix} \\
             \begin{pmatrix}
              0 &0&S^\ast
             \end{pmatrix} &K
            \end{pmatrix}.
\end{equation*}
Let $F\in H_{\rho,-\infty}(\mathbb{R};H_1(\sqrt{\E}A_\Nac\sqrt{\E}+1))$ and let $U$ satisfy 
\begin{equation}\label{eq:gem_null}
 (\partial_0+\sqrt{\E^{-1}}A_\Max\sqrt{\E^{-1}}+\sqrt{\E}A_\Dac\sqrt{\E})U=F.
\end{equation}
Then $U$ satisfies \begin{equation*}(\partial_0+\sqrt{\E^{-1}}A_\Max\sqrt{\E^{-1}}+\sqrt{\E}A_\Dac\sqrt{\E}+\sqrt{\E}A_\Nac\sqrt{\E})U=(1+\partial_0^{-1}\sqrt{\E}A_\Nac\sqrt{\E})F.\end{equation*}
Conversely, if $U$ satisfies  
\begin{equation}\label{eq:ext_maxwell}
(\partial_0+\sqrt{\E^{-1}}A_\Max\sqrt{\E^{-1}}+\sqrt{\E}A_\Dac\sqrt{\E}+\sqrt{\E}A_\Nac\sqrt{\E})U=\tilde{F}
\end{equation}
for some $\tilde{F}\in H_{\rho,-\infty}(\mathbb{R};H_0(\sqrt{\E}A_\Nac\sqrt{\E}+1))$, then $U$ is a solution of (\ref{eq:gem_null}) for $F\coloneqq \partial_0(\partial_0+\sqrt{\E}A_\Nac\sqrt{\E})^{-1}\tilde{F}\in H_{\rho,-\infty}(\mathbb{R};H_1(\sqrt{\E}A_\Nac\sqrt{\E}+1))$. 
\end{thm}

\begin{proof}
 First we note that $\sqrt{\E^{-1}}A_\Max\sqrt{\E^{-1}}, \sqrt{\E}A_\Dac\sqrt{\E}$ and $\sqrt{\E}A_\Nac\sqrt{\E}$ pairwise annihilate each other, due to the given structure of $\E$ (see Remark \ref{rem:extended_maxwell}). Thus, we have that
\begin{multline*}
 (\partial_0+\sqrt{\E}A_\Nac\sqrt{\E})(\partial_0+\sqrt{\E^{-1}}A_\Max\sqrt{\E^{-1}}+\sqrt{\E}A_\Dac\sqrt{\E})\\=\partial_0(\partial_0+\sqrt{\E^{-1}}A_\Max\sqrt{\E^{-1}}+\sqrt{\E}A_\Dac\sqrt{\E}+\sqrt{\E}A_\Nac\sqrt{\E})
\end{multline*}
yielding that
\begin{multline*}
 \partial_0^{-1}(\partial_0+\sqrt{\E^{-1}}A_\Max\sqrt{\E^{-1}}+\sqrt{\E}A_\Dac\sqrt{\E}+\sqrt{\E}A_\Nac\sqrt{\E})^{-1}\\= (\partial_0+\sqrt{\E}A_\Nac\sqrt{\E})^{-1}(\partial_0+\sqrt{\E^{-1}}A_\Max\sqrt{\E^{-1}}+\sqrt{\E}A_\Dac\sqrt{\E})^{-1}.
\end{multline*}Using this factorization result, we may derive the assertion as in the proof of Theorem \ref{thm:connection_max_maxext}.
\end{proof}

\begin{rem}
The system (\ref{eq:gem}) makes its reappearance from a completely different angle
in \cite{zbMATH01334186} and for numerical purposes in \cite{zbMATH05172225},
\cite{Weggler1-2012}, as an alternative to the fully extended Maxwell
system.
\end{rem}

\section{\label{sec:The-Dirac-System}The Dirac System}

>From (\ref{eq:Dirac-factors}) we see that Dirac's problem of finding
a first order factorization of  the differential operator $\square$
of the wave equation, which gave rise to the celebrated Dirac equations
\cite{54.0973.01,54.0973.02}, is solved by the extended Maxwell system
with $\mathcal{E}=1$. This observation will now be utilized to reformulate
Dirac's equation as an extended Maxwell system.

\begin{defn}
We define the differential operator $C(\hat{\partial})$ on $L^2(\mathbb{R}^3)^2$, by\footnote{Note that 
\begin{equation*}
\left(\begin{array}{cc}
\partial_{3} & \partial_{1}-\ii\,\partial_{2}\\
\partial_{1}+\ii\,\partial_{2} & -\partial_{3}
\end{array}\right)
\end{equation*}
is an operator quaternion, since it has the form
\begin{equation*}
\left(\begin{array}{cc}
A & -B^{*}\\
B & A^{*}
\end{array}\right),
\end{equation*}
where $A:D\left(A\right)\subseteq H\to H,\, B:D\left(B\right)\subseteq H\to H$
are closed, densely defined linear operators, such that $A$ has a
non-empty resolvent set $\nu\left(A\right)$ and $A,B^{*}$ are commuting,
i.e.
\begin{equation*}
\left(\lambda-A\right)^{-1}B\subseteq B\left(\lambda-A\right)^{-1}
\end{equation*}
for $\lambda\in\nu\left(A\right)$. If $A,B$ are complex numbers
(as multipliers) this block operator matrix yields a standard representation
of the classical quaternions.%
}
\begin{equation*}
 C(\hat{\partial})\;:=\left(\begin{array}{cc}
\partial_{3} & \partial_{1}-\ii\,\partial_{2}\\
\partial_{1}+\ii\,\partial_{2} & -\partial_{3}
\end{array}\right).
\end{equation*}
 The \emph{Dirac operator} is then given by the block operator matrix on the space $H_{\rho,0}(\mathbb{R};L^2(\mathbb{R}^3)^4)$
\begin{equation*}
 \mathcal{Q}_{0}(\partial_{0},\hat{\partial}) :=  \left(\begin{array}{cc}
\partial_{0}+\ii & C(\hat{\partial})\\
C(\hat{\partial}) & \partial_{0}-\ii
\end{array}\right).
\end{equation*}
\end{defn}

\begin{rem}
Frequently, the operator $C(\hat{\partial})$ is represented by $\sum_{k=1}^{3}\,\Pi_{k}\,\partial_{k},$ where 
\begin{equation*}
\Pi_{1}:=\left(\begin{array}{cc}
0 & +1\\
+1 & 0
\end{array}\right),\:\Pi_{2}:=\left(\begin{array}{cc}
0 & -\ii\\
+\ii & 0
\end{array}\right),\:\Pi_{3}:=\left(\begin{array}{cc}
+1 & 0\\
0 & -1
\end{array}\right)
\end{equation*}
are known as the \emph{Pauli matrices}.
\end{rem}

\begin{lem}
 The Dirac operator $\mathcal{Q}_0(\partial_0,\hat{\partial})$ is unitarily equivalent to the operator
\begin{equation*}
 \mathcal{Q}_1(\partial_0,\hat{\partial})\coloneqq \left(\begin{array}{cc}
\partial_{0} & \ii-\ii\, C(\hat{\partial})\\
\ii+\ii\, C(\hat{\partial}) & \partial_{0}
\end{array}\right).
\end{equation*}
 In particular, $\mathcal{Q}_1$ is of the Hamiltonian form $\partial_0+\begin{pmatrix}
                                                                   0 & -W^\ast\\
                                                                   W& 0
                                                                  \end{pmatrix}$, where $W\coloneqq \i(1+C(\hat{\partial}))$.
\end{lem}

\begin{proof}
 Applying the unitary transformation given by the block matrix $\frac{1}{\sqrt{2}}\begin{pmatrix}
                                                                  +\i & +1 \\
                                                                  +\i & -1
                                                                 \end{pmatrix}$ to $\mathcal{Q}_0(\partial_0,\hat{\partial})$ we obtain
\begin{align*}
 \frac{1}{2}\begin{pmatrix}
                                                                  +\i & +1 \\
                                                                  +\i & -1
                                                                 \end{pmatrix}\left(\begin{array}{cc}
\partial_{0}+\ii & C(\hat{\partial})\\
C(\hat{\partial}) & \partial_{0}-\ii
\end{array}\right)\begin{pmatrix} 
                   -\i & -\i \\
                    1 & -1 
                  \end{pmatrix}&=\frac{1}{2}\begin{pmatrix}
                                                                  +\i & +1 \\
                                                                  +\i & -1
                                                                 \end{pmatrix} \begin{pmatrix}
                                                                                 -\i\partial_0+1+C(\hat{\partial}) & -\i\partial_0+1-C(\hat{\partial}) \\
                                                                                  \partial_0-\i-\i C(\hat{\partial}) & -\partial_0+\i-\i C(\hat{\partial})
                                                                               \end{pmatrix}\\
&= \begin{pmatrix}
    \partial_0 & \i-\i C(\hat{\partial})\\
     \i +\i C(\hat{\partial}) & \partial_0
   \end{pmatrix}\\
&= \mathcal{Q}_1(\partial_0,\hat{\partial}).
\end{align*}
To see that $\mathcal{Q}_1(\partial_0,\hat{\partial})$ has the asserted Hamiltonian structure, we first observe that
\begin{equation*}
 C(\hat{\partial})^\ast=\begin{pmatrix}
                         -\partial_3 & -\partial_1+\i\partial_2\\
                         -\partial_1-\i\partial_2& \partial_3
                        \end{pmatrix}=-C(\hat{\partial})
\end{equation*}
and thus
\begin{equation*}
 W^\ast= -\i-\i C(\hat{\partial})^\ast=-\i+\i C(\hat{\partial}),
\end{equation*}
which yields the assertion.
\end{proof}

To see the connection between the Dirac equation and the extended
Maxwell equation we separate real and imaginary parts. For doing so, we consider $\mathbb{C}$ as a two dimensional Hilbert space over the field $\mathbb{R}$ with inner product $\langle z_0|z_1\rangle\coloneqq \Re z_0^\ast z_1$. Then the mapping
\begin{align*}
 \mathbb{R}^2&\to \mathbb{C} \\
      (x,y) &\mapsto x+\i y
\end{align*}
  becomes unitary. The multiplication with a complex number $a+\i b$, $a,b\in \mathbb{R}$ then becomes the multiplication with the matrix $\begin{pmatrix} 
                                                                                                                      a& -b\\
                                                                                                                      b & a 
                                                                                                                      \end{pmatrix}$.
Consequently, the operator $W$ is unitarily equivalent to 
\begin{equation*}
 \tilde{W}\coloneqq \begin{pmatrix}
  0 & -1-\partial_3 & \partial_2 & -\partial_1 \\
  1+\partial_3 & 0 & \partial_1 & \partial_2 \\
  -\partial_2 & -\partial_1 & 0 & -1+\partial_3\\
   \partial_1 & -\partial_2 & 1-\partial_3 & 0
 \end{pmatrix}= \begin{pmatrix}
  0 & -1 & 0 & 0 \\
  1 & 0 &  0&0 \\
 0 & 0 & 0 & -1\\
  0 & 0 & 1 & 0
 \end{pmatrix} +\begin{pmatrix}
  0 & -\partial_3 & \partial_2 & -\partial_1 \\
  \partial_3 & 0 & \partial_1 & \partial_2 \\
  -\partial_2 & -\partial_1 & 0 & \partial_3\\
   \partial_1 & -\partial_2 & -\partial_3 & 0
 \end{pmatrix}.
\end{equation*}

\begin{thm}
 The operator $\begin{pmatrix}
                0 & -\tilde{W}^\ast \\
               \tilde{W} &0
               \end{pmatrix}$
is unitarily equivalent to 
\begin{equation*}
 \mathcal{M}_1+ \begin{pmatrix}
                     0 & 0 & \dive & 0\\
                     0 & 0 & -\curl & \grad \\
                    \grad & \curl & 0 &0 \\
                     0 & \dive & 0 &0
                    \end{pmatrix},
\end{equation*}
where \[
       \mathcal{M}_{1}=\left(\begin{array}{cc}
\left(\begin{array}{cc}
0 & \left(\begin{array}{ccc}
0 & 0 & 0\end{array}\right)\\
\left(\begin{array}{c}
0\\
0\\
0
\end{array}\right) & \left(\begin{array}{ccc}
0 & 0 & 0\\
0 & 0 & 0\\
0 & 0 & 0
\end{array}\right)
\end{array}\right) & \left(\begin{array}{cc}
\left(\begin{array}{ccc}
0 & 0 & -1\end{array}\right) & 0\\
\left(\begin{array}{ccc}
0 & 1 & 0\\
-1 & 0 & 0\\
0 & 0 & 0
\end{array}\right) & \left(\begin{array}{c}
0\\
0\\
-1
\end{array}\right)
\end{array}\right)\\
\left(\begin{array}{cc}
\left(\begin{array}{c}
0\\
0\\
1
\end{array}\right) & \left(\begin{array}{ccc}
0 & 1 & 0\\
-1 & 0 & 0\\
0 & 0 & 0
\end{array}\right)\\
0 & \left(\begin{array}{ccc}
0 & 0 & 1\end{array}\right)
\end{array}\right) & \left(\begin{array}{cc}
\left(\begin{array}{ccc}
0 & 0 & 0\\
0 & 0 & 0\\
0 & 0 & 0
\end{array}\right) & \left(\begin{array}{c}
0\\
0\\
0
\end{array}\right)\\
\left(\begin{array}{ccc}
0 & 0 & 0\end{array}\right) & 0
\end{array}\right)
\end{array}\right)
      \]

\end{thm}

\begin{proof}
 We compute 
\begin{align*}
 \begin{pmatrix}
  0 & 0 & -1& 0\\
  0 & 0 & 0 & -1\\
 -1 & 0 &0 &0 \\
  0 & 1 &0 &0
 \end{pmatrix}\begin{pmatrix}
  0 & -\partial_3 & \partial_2 & -\partial_1 \\
  \partial_3 & 0 & \partial_1 & \partial_2 \\
  -\partial_2 & -\partial_1 & 0 & \partial_3\\
   \partial_1 & -\partial_2 & -\partial_3 & 0
 \end{pmatrix} \begin{pmatrix}
                 0 & 0 & 0 & 1\\
                 1 & 0 & 0 & 0 \\
                 0 & 1 & 0 &0 \\
                 0 & 0 & 1 & 0
               \end{pmatrix} &= \begin{pmatrix}
  0 & 0 & -1& 0\\
  0 & 0 & 0 & -1\\
 -1 & 0 &0 &0 \\
  0 & 1 &0 &0
 \end{pmatrix} \begin{pmatrix}
               -\partial_3 & \partial_2 & -\partial_1 & 0\\
               0 & \partial_1 & \partial_2 & \partial_3 \\
               -\partial_1 & 0 & \partial_3 & -\partial_2 \\
               -\partial_2 & -\partial_3 & 0 & \partial_1
               \end{pmatrix}\\
 &= \begin{pmatrix}
     \partial_1 & 0 & -\partial_3 & \partial_2 \\
     \partial_2 & \partial_3 & 0 & -\partial_1\\
     \partial_3 & -\partial_2 & \partial _1 & 0 \\
     0 & \partial_1 & \partial_2 & \partial_3
    \end{pmatrix}\\
&= \begin{pmatrix}
    \grad & \curl \\
    0 & \dive
   \end{pmatrix}.
\end{align*}
Moreover, we obtain
\begin{align*}
 \begin{pmatrix}
  0 & 0 & -1& 0\\
  0 & 0 & 0 & -1\\
 -1 & 0 &0 &0 \\
  0 & 1 &0 &0
 \end{pmatrix}\begin{pmatrix}
  0 & -1 & 0 & 0 \\
  1 & 0 & 0 & 0 \\
  0 & 0 & 0 & -1\\
  0 & 0 & 1 & 0
 \end{pmatrix} \begin{pmatrix}
                 0 & 0 & 0 & 1\\
                 1 & 0 & 0 & 0 \\
                 0 & 1 & 0 &0 \\
                 0 & 0 & 1 & 0
               \end{pmatrix} &= \begin{pmatrix}
  0 & 0 & -1& 0\\
  0 & 0 & 0 & -1\\
 -1 & 0 &0 &0 \\
  0 & 1 &0 &0
 \end{pmatrix} \begin{pmatrix}
                -1 & 0 & 0 & 0\\
                0 & 0 & 0 & 1\\
                0 & 0 & -1& 0\\
                0 & 1 & 0 & 0
               \end{pmatrix}\\
 &= \begin{pmatrix}
     0 & 0 & 1 & 0\\
     0 & -1 & 0 & 0\\
     1 & 0 & 0 & 0 \\
    0 & 0 & 0 & 1
    \end{pmatrix}
\end{align*}
and hence, we have that
\begin{equation*}
 \begin{pmatrix}
  0 & 0 & -1& 0\\
  0 & 0 & 0 & -1\\
 -1 & 0 &0 &0 \\
  0 & 1 &0 &0
 \end{pmatrix}\tilde{W}\begin{pmatrix}
                 0 & 0 & 0 & 1\\
                 1 & 0 & 0 & 0 \\
                 0 & 1 & 0 &0 \\
                 0 & 0 & 1 & 0
               \end{pmatrix}=\begin{pmatrix}
     0 & 0 & 1 & 0\\
     0 & -1 & 0 & 0\\
     1 & 0 & 0 & 0 \\
    0 & 0 & 0 & 1
    \end{pmatrix}+\begin{pmatrix}
    \grad & \curl \\
    0 & \dive
   \end{pmatrix}.
\end{equation*}
Consequently, 
\begin{equation*}
 \begin{pmatrix}
                 0 & 1 & 0 & 0\\
                 0 & 0 & 1 & 0 \\
                 0 & 0 & 0 &1 \\
                 1 & 0 & 0 & 0
               \end{pmatrix}\tilde{W}^\ast \begin{pmatrix}
  0 & 0 & -1& 0\\
  0 & 0 & 0 & 1\\
 -1 & 0 &0 &0 \\
  0 & -1 &0 &0
 \end{pmatrix}=\begin{pmatrix}
     0 & 0 & 1 & 0\\
     0 & -1 & 0 & 0\\
     1 & 0 & 0 & 0 \\
    0 & 0 & 0 & 1
    \end{pmatrix}+\begin{pmatrix}
    -\dive & 0 \\
    \curl & -\grad
   \end{pmatrix}
\end{equation*}
which gives that
\begin{align*}
& \left(\begin{array}{cc}
\left(\begin{array}{cccc}
0 & 1 & 0 & 0\\
0 & 0 & 1 & 0\\
0 & 0 & 0 & 1\\
1 & 0 & 0 & 0
\end{array}\right) & \left(\begin{array}{cccc}
0 & 0 & 0 & 0\\
0 & 0 & 0 & 0\\
0 & 0 & 0 & 0\\
0 & 0 & 0 & 0
\end{array}\right)\\
\left(\begin{array}{cccc}
0 & 0 & 0 & 0\\
0 & 0 & 0 & 0\\
0 & 0 & 0 & 0\\
0 & 0 & 0 & 0
\end{array}\right) & \left(\begin{array}{cccc}
0 & 0 & -1 & 0\\
0 & 0 & 0 & -1\\
-1 & 0 & 0 & 0\\
0 & 1 & 0 & 0
\end{array}\right)
\end{array}\right)\left(\begin{array}{cc}
0 & -W^{*}\\
W & 0
\end{array}\right)\left(\begin{array}{cc}
\left(\begin{array}{cccc}
0 & 0 & 0 & 1\\
1 & 0 & 0 & 0\\
0 & 1 & 0 & 0\\
0 & 0 & 1 & 0
\end{array}\right) & \left(\begin{array}{cccc}
0 & 0 & 0 & 0\\
0 & 0 & 0 & 0\\
0 & 0 & 0 & 0\\
0 & 0 & 0 & 0
\end{array}\right)\\
\left(\begin{array}{cccc}
0 & 0 & 0 & 0\\
0 & 0 & 0 & 0\\
0 & 0 & 0 & 0\\
0 & 0 & 0 & 0
\end{array}\right) & \left(\begin{array}{cccc}
0 & 0 & -1 & 0\\
0 & 0 & 0 & 1\\
-1 & 0 & 0 & 0\\
0 & -1 & 0 & 0
\end{array}\right)
\end{array}\right)=\\
 & =\left(\begin{array}{cc}
\left(\begin{array}{cc}
0 & \left(\begin{array}{ccc}
0 & 0 & 0\end{array}\right)\\
\left(\begin{array}{c}
0\\
0\\
0
\end{array}\right) & \left(\begin{array}{ccc}
0 & 0 & 0\\
0 & 0 & 0\\
0 & 0 & 0
\end{array}\right)
\end{array}\right) & \left(\begin{array}{cc}
\dive & 0\\
-\curl & \grad
\end{array}\right)+\left(\begin{array}{cc}
\left(\begin{array}{ccc}
0 & 0 & -1\end{array}\right) & 0\\
\left(\begin{array}{ccc}
0 & 1 & 0\\
-1 & 0 & 0\\
0 & 0 & 0
\end{array}\right) & \left(\begin{array}{c}
0\\
0\\
-1
\end{array}\right)
\end{array}\right)\\
\left(\begin{array}{cc}
\grad & \curl\\
0 & \dive
\end{array}\right)+\left(\begin{array}{cc}
\left(\begin{array}{c}
0\\
0\\
1
\end{array}\right) & \left(\begin{array}{ccc}
0 & 1 & 0\\
-1 & 0 & 0\\
0 & 0 & 0
\end{array}\right)\\
0 & \left(\begin{array}{ccc}
0 & 0 & 1\end{array}\right)
\end{array}\right) & \left(\begin{array}{cc}
\left(\begin{array}{ccc}
0 & 0 & 0\\
0 & 0 & 0\\
0 & 0 & 0
\end{array}\right) & \left(\begin{array}{c}
0\\
0\\
0
\end{array}\right)\\
\left(\begin{array}{ccc}
0 & 0 & 0\end{array}\right) & 0
\end{array}\right)
\end{array}\right).
\end{align*} 
Since this transformation is obviously unitary, we derive the assertion.
\end{proof}

\begin{rem}$\,$\label{rem:Dirac}
\begin{enumerate}[(a)] 
\item We have shown that in the free-space situation, the Dirac operator is unitarily equivalent to the extended Maxwell operator\footnote{In the framework of quaternions a connection between a differently extended time-harmonic Maxwell operator and the time-harmonic Dirac
operator has earlier been discovered by Kravchenko and Shapiro, \cite{0305-4470-28-17-030-1},
compare also \cite{Krav}.
} (here
in the Hamiltonian form, see Remark \ref{rem:extended_maxwell}) 
\begin{equation*}
\partial_{0}+\mathcal{M}_{1}+\left(\begin{array}{cccc}
0 & 0 & \dive & 0\\
0 & 0 & -\curl & \grad\\
\grad & \curl & 0 & 0\\
0 & \dive & 0 & 0
\end{array}\right),
\end{equation*}
 considered as operators on the real Hilbert space $H_{\rho,0}(\mathbb{R};L^2(\mathbb{R}^3;\mathbb{R}^2)^4)$. Here, 
\begin{equation*}
\mathcal{M}_{1}=\left(\begin{array}{cc}
\left(\begin{array}{cc}
0 & \left(\begin{array}{ccc}
0 & 0 & 0\end{array}\right)\\
\left(\begin{array}{c}
0\\
0\\
0
\end{array}\right) & \left(\begin{array}{ccc}
0 & 0 & 0\\
0 & 0 & 0\\
0 & 0 & 0
\end{array}\right)
\end{array}\right) & \left(\begin{array}{cc}
\left(\begin{array}{ccc}
0 & 0 & -1\end{array}\right) & 0\\
\left(\begin{array}{ccc}
0 & 1 & 0\\
-1 & 0 & 0\\
0 & 0 & 0
\end{array}\right) & \left(\begin{array}{c}
0\\
0\\
-1
\end{array}\right)
\end{array}\right)\\
\left(\begin{array}{cc}
\left(\begin{array}{c}
0\\
0\\
1
\end{array}\right) & \left(\begin{array}{ccc}
0 & 1 & 0\\
-1 & 0 & 0\\
0 & 0 & 0
\end{array}\right)\\
0 & \left(\begin{array}{ccc}
0 & 0 & 1\end{array}\right)
\end{array}\right) & \left(\begin{array}{cc}
\left(\begin{array}{ccc}
0 & 0 & 0\\
0 & 0 & 0\\
0 & 0 & 0
\end{array}\right) & \left(\begin{array}{c}
0\\
0\\
0
\end{array}\right)\\
\left(\begin{array}{ccc}
0 & 0 & 0\end{array}\right) & 0
\end{array}\right)
\end{array}\right)                                                                                                                    
\end{equation*}
 is skew-selfadjoint, i.e. from the electrodynamics perspective we
are in a chiral media case. 
\item It is a rather remarkable observation that the Dirac equation is so
closely connected to the extended Maxwell system. It appears from
this perspective that spinors are actually a redundant construction
since the alternating forms setup for the extended Maxwell system
is already quite sufficient to discuss Dirac's equation. The interpretation
of this observation is not a mathematical issue but may well be a
matter for theoretical physicists to contemplate. Due to unitary equivalence
of course nothing of the spinor structure is actually lost and may
still be utilized if desired. Since the extended Maxwell system is
actually a differential forms construction, as a by-product we get
a natural transfer to Riemannian manifolds other than $\mathbb{R}^{3}$
for free, see Remark \ref{rem:extended_maxwell} (d). As
for the consideration of lower dimensional cases, we recall the discussion
in Remark \ref{rem:extended_maxwell} (d). 
\end{enumerate}
\end{rem}

\section{\label{sec:The-Maxwell-Dirac-System}The Maxwell-Dirac System}
The Maxwell-Dirac System (\cite{zbMATH05724791}) is a coupled system of Maxwell's equations and the Dirac equation. The coupling occurs via the so-called vector and scalar potentials. In Physics literature solutions to the Dirac equation are also known as the spinors. Having demonstrated the connections of the Maxwell equation to the extended Maxwell system as well as the connection of the Dirac equation to the extended Maxwell system, it remains to provide the relationship of the vector and scalar potentials to the extended Maxwell system\footnote{In this context it is remarkable that already Dirac, by suggesting
to include magnetic charge densities in Maxwell's equations, \cite{Dirac01091931}
, was led to consider a potential of the general form $\left(\begin{array}{c}
\alpha_{0}\\
\alpha_{1}\\
\alpha_{2}\\
\alpha_{3}
\end{array}\right)$, which in effect leads to an extended Maxwell system not for the
electromagnetic field but for its potential: 
\begin{equation*}
\left(\partial_{0}-A\right)\alpha=\left(\begin{array}{c}
0\\
E\\
H\\
0
\end{array}\right)+\delta\otimes\alpha_{\left(0\right)}.
\end{equation*}
}. This will actually play the major role in this section. For simplicity, we focus on the case, where the material parameters are of the form $\mathcal{E}=1$ and $M_0=1$. We shall abbreviate
\begin{equation}\label{eq:operator}
A=\left(A_{\mathrm{Max}}+A_{\mathrm{Dac}}+A_{\mathrm{Nac}}\right)=\left(\begin{array}{cccc}
0 & \dive & 0 & 0\\
\interior\grad & 0 & -\curl & 0\\
0 & \interior\curl & 0 & \grad\\
0 & 0 & \interior\dive & 0
\end{array}\right)
\end{equation}
on the open set $\Omega\subseteq \mathbb{R}^3$. The relationship of the equations for the scalar and the vector potentials to the extended Maxwell system is as follows:
\begin{thm}\label{thm:pot_solves_Dirac} Let $\Omega\subseteq \mathbb{R}^3$ open and let $(E,H)\in H_{\rho,-\infty}(\mathbb{R},L^2(\Omega)^6)$ satisfy Maxwell's equation, i.e.,
\[
   \partial_0 \begin{pmatrix}
                 E \\ H
              \end{pmatrix}+ \begin{pmatrix}
                 0 & -\curl \\ \interior \curl & 0
              \end{pmatrix}\begin{pmatrix} E\\ H\end{pmatrix} = \begin{pmatrix}
                 -J \\ 0
              \end{pmatrix}+\delta\otimes \begin{pmatrix}
                 E_0 \\ H_0
              \end{pmatrix}
\]
with the property that $E_0\in D(\dive)$, $H_0 \in \interior{\curl}[D(\interior{\curl})]$, and $J\in H_{\rho,0}(\mathbb{R};L^2(\Omega)^3)$ such that $J=0$ on $]-\infty,0[$.
Let $\alpha=(\alpha_0,\alpha_1,\alpha_2,\alpha_3)\in H_{\rho,-\infty}(\mathbb{R};L^2(\Omega)\oplus L^2(\Omega)^3\oplus L^2(\Omega)^3\oplus L^2(\Omega))$, $\alpha_{1,0}\in D(\interior{\curl})$. Then the following statements are equivalent:
\begin{enumerate}[(i)]
 \item $\alpha$ satisfies \[
                           (\partial_0 - A)\alpha=\begin{pmatrix}
                                                   0\\E\\H\\0
                                                  \end{pmatrix}+ \delta\otimes\begin{pmatrix}
                                                   0\\ \alpha_{1,0}\\ 0\\0
                                                  \end{pmatrix}
                          \]
  and $H_0 =-\interior{\curl}\alpha_{1,0}$,
 \item $\alpha_2=0$, $\alpha_3=0$, and $\alpha_0$, $\alpha_1$ satisfy
\[
   E=\partial_0 (\alpha_1- \chi_{]0,\infty[}\otimes \alpha_{1,0})- \grad \alpha_0\text{ and } H= -\interior{\curl}\alpha_1.
\]
\end{enumerate}
\end{thm}
\begin{proof}
 Let $H\coloneqq L^2(\Omega)\oplus L^2(\Omega)^3\oplus L^2(\Omega)^3\oplus L^2(\Omega)$. Assume that $\alpha$ satisfies the extended Maxwell system with right hand side $\begin{pmatrix}
                                                   0\\E\\H\\0
                                                  \end{pmatrix}+ \delta\otimes\begin{pmatrix}
                                                   0\\ \alpha_{1,0}\\ 0\\0
                                                  \end{pmatrix}$. 
We have in $H_{\nu,-\infty}\left(\mathbb{R},H\right)$
\begin{align*}
\partial_{0}\alpha_{0}-\dive\alpha_{1} & =0,\\
\partial_{0}\left(\alpha_{1}-\chi_{_{\oi0\infty}}\otimes\alpha_{1,0}\right)-\interior\grad\alpha_{0}+\curl\alpha_{2} & =E,\\
\partial_{0}\alpha_{2}-\interior\curl\alpha_{1}-\grad\alpha_{3} & =H,\\
\partial_{0}\alpha_{3}-\interior\dive\alpha_{2} & =0.
\end{align*}
We consider the last two equations
\begin{align*}
\partial_{0}\alpha_{2}-\interior\curl\alpha_{1}-\grad\alpha_{3} & =H,\\
\partial_{0}\alpha_{3}-\interior\dive\alpha_{2} & =0.
\end{align*}
Applying $\interior\dive$ to the first and substituting $\interior\dive\alpha_{2}$
from the second equation yields, still computing in the space $H_{\nu,-\infty}\left(\mathbb{R},H\right)$,
\begin{equation*}
\partial_{0}^{2}\alpha_{3}-\interior\dive\grad\alpha_{3}  =\interior\dive H.
\end{equation*}
Since $H$ satisfies the equation
\begin{equation*}
 \partial_0H+\interior\curl E =\delta\otimes H_0,
\end{equation*}
we obtain, using that $H_0\in \interior\curl[D(\interior\curl)]$,  
\begin{equation}\label{eq:H_in_curl}
H\in H_{\rho,-\infty}\left(\mathbb{R};\overline{\interior\curl[D(\interior\curl)]}\right)\subseteq H_{\rho,-\infty}(\mathbb{R};N(\interior{\dive})). 
\end{equation}
Consequently, $\alpha_3$ satisfies the wave equation with vanishing source term from which we derive $\alpha_3=0$. The latter gives
\begin{equation*}
\interior\dive\alpha_{2}=0.
\end{equation*}
 Consider now the second equation
\begin{equation}
\partial_{0}\left(\alpha_{1}-\chi_{_{\oi0\infty}}\otimes\alpha_{1,0}\right)-\interior\grad\alpha_{0}+\curl\alpha_{2}=E.\label{eq:pot-2}
\end{equation}
Since $E,H$ satisfy Maxwell's equations we have in particular $\interior\curl E\in H_{\nu,-\infty}\left(\mathbb{R},H\right)$
as well as $\interior\curl\alpha_{1}\in H_{\nu,-\infty}\left(\mathbb{R},H\right)$
and so recalling that 
\begin{equation*}
\alpha_{1,0}\in D\left(\interior\curl\right)
\end{equation*}
and applying $\interior\curl$ to (\ref{eq:pot-2}) we get using one
of the Maxwell equations 
\begin{align*}
\partial_{0}\interior\curl\alpha_{1}+\interior\curl\left(\curl\alpha_{2}\right) & =\interior\curl E+\delta\otimes\interior\curl\alpha_{1,0}\\
 & =-\partial_{0}H+\delta\otimes\left(H_{0}+\interior\curl\alpha_{1,0}\right).
\end{align*}
Thus, 
\begin{align}
\partial_{0}\left(H+\interior\curl\alpha_{1}\right)+\interior\curl\left(\curl\alpha_{2}\right) & =\delta\otimes\left(H_{0}+\interior\curl\alpha_{1,0}\right)=0,\label{eq:alpha12-1}
\end{align}
by our assumption for the initial data $H_{0},\:\alpha_{1,0}$. Now, from the third equation we obtain that 
\begin{align}
\partial_{0}\alpha_{2}-\interior\curl\alpha_{1} & =H,\nonumber \\
\partial_{0}\left(\curl\alpha_{2}\right)-\curl\left(\interior\curl\alpha_{1}\right) & =\curl H,\nonumber \\
\partial_{0}\left(\curl\alpha_{2}\right)-\curl\left(H+\interior\curl\alpha_{1}\right) & =0.\label{eq:alpha21}
\end{align}
With (\ref{eq:alpha12-1}) we have that $\left(\begin{array}{c}
\curl\alpha_{2}\\
H+\interior\curl\alpha_{1}
\end{array}\right)$ satisfies a Maxwell system for vanishing data. Thus, we have also
$\curl\alpha_{2}=0$ and $H+\interior\curl\alpha_{1}=0$. Since we have that $\alpha_{2}\in\overline{\interior\curl\left[D\left(\interior\curl\right)\right]}$, by (\ref{eq:H_in_curl}) and the first equation of (\ref{eq:alpha21}), and $\curl\alpha_{2}=0$, we get $\alpha_{2}=0$, which shows the implication (i) $\Rightarrow$ (ii).
The converse implication is trivial.
\end{proof}
\begin{rem}
The tempted reader might wonder that for the potential $\alpha$ being the unique solution of a partial differential equations, there is no room left for the issue of gauge invariance. This is, however, not the case. In fact, we observe that different potentials result in a different right-hand side of the partial differential equation under consideration.
 More precisely, we have
\begin{equation*}
\left(\partial_{0}-A\right)\left(\begin{array}{c}
\alpha_{0}+\partial_{0}\varphi\\
\alpha_{1}+\interior\grad\varphi\\
0\\
0
\end{array}\right)=\left(\begin{array}{c}
\left(\partial_{0}^{2}-\dive\interior\grad\right)\varphi\\
E\\
H\\
0
\end{array}\right)+\delta\otimes\left(\begin{array}{c}
0\\
\alpha_{1,0}\\
0\\
0
\end{array}\right)
\end{equation*}
for suitable scalar fields $\varphi$.
\end{rem}

Now, we come back to the original Maxwell-Dirac system. As it has been demonstrated in the previous sections, the spinor part for the banded extended Maxwell version (i.e. the part of the Dirac equation) corresponds to\footnote{Note that in Remark \ref{rem:Dirac} the operator $\mathcal{M}_1$ corresponds to the Hamiltonian form of the extended Maxwell system, while here we deal with the non-Hamiltonian form (\ref{eq:operator}).}
\begin{align*}
M_{1} & =\left(\begin{array}{cccccccc}
0 & 0 & 0 & 1 & 0 & 0 & 0 & 0\\
0 & 0 & 0 & 0 & 0 & 1 & 0 & 0\\
0 & 0 & 0 & 0 & -1 & 0 & 0 & 0\\
-1 & 0 & 0 & 0 & 0 & 0 & 0 & 0\\
0 & 0 & 1 & 0 & 0 & 0 & 0 & 0\\
0 & -1 & 0 & 0 & 0 & 0 & 0 & 0\\
0 & 0 & 0 & 0 & 0 & 0 & 0 & 1\\
0 & 0 & 0 & 0 & 0 & 0 & -1 & 0
\end{array}\right)\\
 & =\left(\begin{array}{cccc}
0 & \left(\begin{array}{ccc}
0 & 0 & 1\end{array}\right) & \left(\begin{array}{ccc}
0 & 0 & 0\end{array}\right) & 0\\
\left(\begin{array}{c}
0\\
0\\
-1
\end{array}\right) & \left(\begin{array}{ccc}
0 & 0 & 0\\
0 & 0 & 0\\
0 & 0 & 0
\end{array}\right) & \left(\begin{array}{ccc}
0 & 1 & 0\\
-1 & 0 & 0\\
0 & 0 & 0
\end{array}\right) & \left(\begin{array}{c}
0\\
0\\
0
\end{array}\right)\\
\left(\begin{array}{c}
0\\
0\\
0
\end{array}\right) & \left(\begin{array}{ccc}
0 & 1 & 0\\
-1 & 0 & 0\\
0 & 0 & 0
\end{array}\right) & \left(\begin{array}{ccc}
0 & 0 & 0\\
0 & 0 & 0\\
0 & 0 & 0
\end{array}\right) & \left(\begin{array}{c}
0\\
0\\
1
\end{array}\right)\\
0 & \left(\begin{array}{ccc}
0 & 0 & 0\end{array}\right) & \left(\begin{array}{ccc}
0 & 0 & -1\end{array}\right) & 0
\end{array}\right).
\end{align*}
Moreover, recall from Theorem \ref{thm:pot_solves_Dirac} that under the constraint
 \begin{equation}
H_{0}+\interior\curl\alpha_{1,0}=0\label{eq:pot-req-3}
\end{equation}
 the potential $\alpha$ is the solution of 
\begin{equation*}
\left(\partial_{0}-A\right)\alpha=\left(\begin{array}{c}
0\\
E\\
H\\
0
\end{array}\right)+\delta\otimes\left(\begin{array}{c}
0\\
\alpha_{1,0}\\
0\\
0
\end{array}\right)\in H_{\nu,0}\left(\mathbb{R},H\right),
\end{equation*}
where $\left(\begin{array}{c}
0\\
E\\
H\\
0
\end{array}\right)$ is the solution $U^{\left(0\right)}$ of 
\begin{equation*}
\left(\partial_{0}+A\right)U^{\left(0\right)}=\left(\begin{array}{c}
\rho\\
-J\\
0\\
0
\end{array}\right)+\delta\otimes\left(\begin{array}{c}
0\\
E_{0}\\
H_{0}\\
0
\end{array}\right),
\end{equation*}
which  requires the compatibility condition 
\begin{equation}
\rho=-\dive\partial_{0}^{-1}J+\chi_{_{\oi0\infty}}\otimes\dive E_{0}. \label{eq:compatible}
\end{equation}
So the Maxwell-Dirac system is of the form
\begin{align*}
 & \left(\partial_{0}+\tilde{M}_{1}+\left(\begin{array}{ccc}
A & 0 & 0\\
0 & A & 0\\
0 & 0 & -A
\end{array}\right)\right)U=\\
 & =\left(\begin{array}{c}
\left(\begin{array}{c}
\rho\\
-J\\
0\\
0
\end{array}\right)+\delta\otimes\left(\begin{array}{c}
0\\
E_{0}\\
H_{0}\\
0
\end{array}\right)\\
g+\delta\otimes\psi_{0}\\
\left(\partial_{0}+A\right)^{-1}\left(\left(\begin{array}{c}
\rho\\
-J\\
0\\
0
\end{array}\right)+\delta\otimes\left(\begin{array}{c}
0\\
E_{0}\\
H_{0}\\
0
\end{array}\right)\right)+\delta\otimes\left(\begin{array}{c}
0\\
\alpha_{1,0}\\
0\\
0
\end{array}\right)
\end{array}\right)\eqqcolon F,
\end{align*}
where under assumptions (\ref{eq:compatible}), (\ref{eq:pot-req-3}),
$U$ is of the form $U=\left(\begin{array}{c}
\left(\begin{array}{c}
0\\
E\\
H\\
0
\end{array}\right)\\
\psi\\
\left(\begin{array}{c}
\alpha_{0}\\
\alpha_{1}\\
0\\
0
\end{array}\right)
\end{array}\right)$. 

There is no apparent coupling in this system. In the actual Maxwell-Dirac
system the coupling is via quadratic non-linearities built into the
right-hand side $J,g$. The real-valued eight-component $\psi$ corresponds
to the four-component complex Dirac state (up to a (real) unitary
mapping) for which usually the same name is used. We have
\begin{equation*}
\tilde{M}_{1}=\left(\begin{array}{ccc}
0 & 0 & 0\\
0 & M_{1} & 0\\
0 & 0 & 0
\end{array}\right),
\end{equation*}
with $M_{1}$ as above ($M_{1}$ is skew-selfadjoint).
\begin{rem}
The mentioned non-linearities have to at least match the structure
for a right-hand side capable of satisfying the original Maxwell equations,
see e.g. \cite{zbMATH05679674,zbMATH05724791}. 

Let $A_{k}\in\mathbb{R}^{8\times8}$ be selfadjoint matrices, such that
\begin{equation*}
A=\sum_{k=1}^{3}A_{k}\partial_{k}.
\end{equation*}
Then we set
\begin{align}
J&=\left(\left\langle \psi|A_{k}\psi\right\rangle \right)_{k}\label{eq:MD-req1},\\
g &=\alpha_{0}S\psi+\sum_{k=1}^3\alpha_{k}SA_{k}\psi,\label{eq:MD-req-2},\\
SA_{k}&=A_{k}S \nonumber
\end{align}
with $S$ skew-selfadjoint, $\alpha_k\in \mathbb{R}$ for $k\in \{1,2,3\}$. 

In this case, (on a purely formal level) we need to have for real-valued data
and solutions according to (\ref{eq:compatible}):
\begin{align*}
\partial_{0}\rho & =-\dive J+\delta\otimes\dive E_{0}.
\end{align*}
Formally calculating the right-hand side we get (where we compute in the real setting, i.e. $\mathbb{R}$ is the underlying scalar field)
\begin{align*}
-\dive J+\delta\otimes\dive E_{0} & =-2\left\langle \psi|A_{k}\partial_{k}\psi\right\rangle +\delta\otimes\dive E_{0}\\
 & =\partial_{0}\left|\psi\right|^{2}-2\left\langle \psi|\left(\partial_{0}+A\right)\psi\right\rangle +\delta\otimes\dive E_{0}\\
 & =\partial_{0}\left|\psi\right|^{2}+2\left\langle \psi|\left(M_{1}-\alpha_{0}S-\sum_{k=1}^3\alpha_{k}SA_{k}\right)\psi-\delta\otimes\psi_{0}\right\rangle \\ &\quad +\delta\otimes\dive E_{0}\\
 & =\partial_{0}\left|\psi\right|^{2}+2\left\langle \psi|\left(M_{1}-\alpha_{0}S-\sum_{k=1}^3\alpha_{k}SA_{k}\right)\psi\right\rangle -\delta\otimes\left\langle \psi_{0}|\psi_{0}\right\rangle\\ &\quad  +\delta\otimes\dive E_{0}\\
 & =\partial_{0}\left|\psi\right|^{2}+2\left\langle \psi|\left(M_{1}-\alpha_{0}S-\sum_{k=1}^3\alpha_{k}SA_{k}\right)\psi\right\rangle \\
 & =\partial_{0}\left|\psi\right|^{2}+2\left\langle \psi|M_{1}\psi\right\rangle -2\alpha_{0}\left\langle \psi|S\psi\right\rangle -\sum_{k=1}^32\alpha_{k}\left\langle \psi|SA_{k}\psi\right\rangle +\\
 & \quad -\delta\otimes\left(\left|\psi_{0}\right|^{2}-\dive E_{0}\right)\\
 & =\partial_{0}\left|\psi\right|^{2}-\delta\otimes\left(\left|\psi_{0}\right|^{2}-\dive E_{0}\right).
\end{align*}
So, assuming 
\begin{equation}
\left|\psi_{0}\right|^{2}=\dive E_{0}\label{eq:MD-req3}
\end{equation}
 we have that the appropriate charge density coupling term matching
the assumed quadratic non-linearity (\ref{eq:MD-req1}) must be 
\begin{equation*}
\rho=\left|\psi\right|^{2}
\end{equation*}
 to ensure that the current density $J$ is suitable to generate an
electro-magnetic field.

\end{rem}

\end{document}